\documentclass[11pt,letterpaper]{article}

\usepackage{amsmath,amsthm,amsfonts,amssymb}
\usepackage{thm-restate}
\usepackage{fullpage}
\usepackage[utf8]{inputenc}
\usepackage[dvipsnames]{xcolor}
\usepackage{xspace,enumerate}
\usepackage[shortlabels]{enumitem}
\usepackage[hypertexnames=false,colorlinks=true,urlcolor=Blue,citecolor=Green,linkcolor=BrickRed]{hyperref}
\usepackage[capitalise]{cleveref}
\usepackage[OT4]{fontenc}
\usepackage[ruled]{algorithm}
\usepackage[noend]{algorithmic}
\usepackage{ifpdf}
\usepackage{todonotes}
\usepackage{enumerate}
\usepackage{authblk}
\usepackage{thmtools}

\usepackage[margin=1in]{geometry}
\title{Fully Dynamic Shortest Paths in Sparse Digraphs\thanks{This work has been partially supported by the ERC CoG grant TUgbOAT no 772346.}
\footnote{The initial version of this paper has been published in ICALP 2023~\cite{KarczmarzS23} and claimed, apart from the main fully dynamic APSP result, a secondary result on fully dynamic reachability for \emph{general} sparse digraphs: a data structure with $\widetilde{O}(n\sqrt{m})$ worst-case update time and $O(\sqrt{m})$ query time. Unfortunately, this latter result, and in particular the statement we made about sparse matrix inverse on the way, is incorrect. We are grateful to Jan van den Brand (personal communication) for pointing out the error to us. We discuss the issue in more detail in Section~\ref{sec:error}. It is worth noting that the approach we used still works for DAGs (see Section~\ref{sec:dag}). }}

\date{\vspace{-5ex}}
\author[1]{Adam Karczmarz}
\author[1]{Piotr Sankowski}
\affil[1]{Institute of Informatics, University of Warsaw and IDEAS NCBR, Poland}

\theoremstyle{plain}
\newtheorem{theorem}{Theorem}[section]
\newtheorem{lemma}[theorem]{Lemma}

\newtheorem{remark}[theorem]{Remark}

\def\poly{\operatorname{poly}}

\usepackage{todonotes}

\begin{document}

\maketitle

\newcommand{\Ot}{\ensuremath{\widetilde{O}}}
\newcommand{\eps}{\ensuremath{\epsilon}}
\newcommand{\dist}{\delta}
\newcommand{\len}{\ell}
\newcommand{\wei}{w}
\newcommand{\field}{\mathbb{F}}

\begin{abstract}
  We study the exact fully dynamic shortest paths problem. For real-weighted directed graphs, we show
  a deterministic fully dynamic data structure with $\Ot(mn^{4/5})$ worst-case update time processing arbitrary $s,t$-distance queries in
  $\Ot(n^{4/5})$ time.
  This constitutes the first non-trivial update/query tradeoff for this problem in the regime of \emph{sparse weighted} directed graphs.
\end{abstract}

\section{Introduction}
Computing
all-pairs shortest paths (APSP) is among
the most fundamental algorithmic problems on directed graphs.
This classical problem is often generalized into a data structure ``oracle'' variant: given a graph~$G$, preprocess
$G$ so that efficient point-to-point distance or shortest paths queries are supported.
Computing APSP can be viewed as an extreme solution to the oracle variant;
if one precomputes the answers to all the $n^2$ possible queries in $\Ot(nm)$ time, the queries can be answered in constant time.
The other extreme solution is to not preprocess $G$ at all and run near-linear-time Dijkstra's algorithm
upon each query.
Interestingly, for general directed weighted graphs, no other tradeoffs for the exact oracle variant of static APSP beyond these trivial ones are known.

In this paper, we consider the exact APSP problem, and its easier relative \emph{all-pairs reachability} (or, in other words, \emph{transitive closure}), in the \emph{fully dynamic} setting,
where the input graph $G$ evolves by both edge insertions and deletions.

\subsection{Prior work}
There has been extensive previous work on APSP and transitive closure in the fully dynamic setting.
Notably,
Demetrescu~and~Italiano~\cite{DemetrescuI04}
showed that APSP in a real-weighted digraph can be maintained deterministically
in $\Ot(n^2)$ amortized time
per vertex update (changing all edges incident to a single vertex).
Thorup~\cite{Thorup04}
later slightly improved and simplified their result.
These data structures maintain an explicit distance matrix and the corresponding collection of shortest paths, and thus
allow querying distances and shortest paths in optimal time.
Similar amortized bounds have been earlier obtained for transitive closure~\cite{DemetrescuI05, King99, Roditty08} albeit
using different combinatorial techniques.
Polynomially worse (but nevertheless subcubic) \emph{worst-case} update bounds for real-weighted fully dynamic APSP are also known:
randomized $\Ot(n^{2+1/2})$~\cite{Mao24a} and deterministic $\Ot(n^{2+41/61})$~\cite{ChechikZ23}.

For dense \emph{unweighted} digraphs, non-trivial fully dynamic data structures for all-pairs reachability and APSP
can be obtained using algebraic techniques. Via a reduction to dynamic matrix inverse,
Sankowski~\cite{Sankowski04}
obtained $O(n^2)$ \emph{worst-case} update bound
for explicitly maintaining the transitive closure, and also gave update/query tradeoffs.
In particular,~he showed a reachability data structure with subquadratic
$O(n^{1.529})$ update time and sublinear $O(n^{0.529})$ query time.
Using the same general algebraic framework,
van den Brand, Nanongkai, and Saranurak~\cite{BrandNS19}
showed $O(n^{1.407})$ worst-case update
bound for $st$-reachability (that is, fixed single-pair reachability), whereas
van den Brand, Forster, and Nazari~\cite{BrandFN22} gave an $O(n^{1.704})$ worst-case update
bound for maintaining exact $st$-distance in unweighted digraphs.\footnote{The single-pair data structures~\cite{BrandNS19, BrandFN22} can be easily extended to support arbitrary-pair queries. Then, the query time matches the update time.}  That framework, however, inherently leads to Monte Carlo randomized solutions and does
not generally allow reporting (shortest) paths within the stated query bounds.\footnote{As shown quite recently, reporting (shortest) paths in subquadratic time
can be possible via a combination of algebraic and combinatorial techniques~\cite{AlokhinaB24, BergamaschiHGWW21, Karczmarz0S22}. However, this comes with a polynomial time overhead.}

Interestingly, neither the known fully dynamic APSP data structures for \emph{real-weighted} digraphs (or even for integer weights between $1$ and $n$) nor
the algebraic data structures tailored to dense graphs yield any improvement over the extreme
recompute-from-scratch approaches for \emph{sparse} graphs with $m=\Ot(n)$. This is especially unfortunate as such graphs
are ubiquitous in real-world applications.
Indeed, for ${m=\Ot(n)}$, recomputing APSP from scratch takes $\Ot(n^2)$ worst-case update time and $O(1)$ query time (which matches the amortized bound in~\cite{DemetrescuI05,Thorup04}),
whereas naively running Dijkstra's algorithm upon query costs $\Ot(n)$ time (which already improves upon the update bound of the algebraic $st$-distance data structure of~\cite{BrandFN22}).
The only non-trivial fully dynamic APSP data structure in the sparse regime has been
described by
Roditty and Zwick~\cite{RodittyZ11}. Their randomized data structure has $\Ot(m\sqrt{n})$ amortized update
time and $O(n^{3/4})$ query time. Unfortunately, it works only for \emph{unweighted digraphs}.
To the best of our knowledge, no non-trivial update/query tradeoffs for fully dynamic
APSP in sparse weighted digraphs have been described to date.
A step towards this direction has been made by
Karczmarz~\cite{Karczmarz21}
who showed that some \emph{fixed} -- in a crucial way -- $m$
distance pairs can be maintained in $\Ot(mn^{2/3})$ worst-case time per update.

For the simpler fully dynamic reachability problem, the $O(n^{1.529})$ update time and $O(n^{0.529})$ query time
algebraic tradeoff of~\cite{Sankowski04} is already non-trivial for all graph densities.
However, specifically for sparse graphs, a deterministic and combinatorial tradeoff of Roditty~and~Zwick~\cite{RodittyZ08}
is more efficient; they showed a data structure with $O(m\sqrt{n})$ amortized update time
and $O(\sqrt{n})$ query time. Moreover, the data structure of~\cite{Sankowski04} requires fast matrix multiplication algorithms~\cite{WilliamsXXZ24}
and these are considered impractical. That being said, the downside
of~\cite{RodittyZ08} is that the update bound holds only in the amortized sense.

\subsection{Our result}
We show the first fully dynamic APSP data structure with
non-trivial update and query bounds for \emph{sparse weighted} digraphs.

\begin{restatable}{theorem}{tshpath}\label{t:shpath}
  Let $G$ be a real-weighted directed graph. There exists a deterministic data structure maintaining $G$
  under fully dynamic vertex updates and answering arbitrary $s,t$-distance queries with $\Ot(mn^{4/5})$ worst-case update
  time and $\Ot(n^{4/5})$ query time and using $\Ot(n^2)$ space. The queries are supported only when $G$ has no negative cycles.
  After answering~a distance query, some corresponding shortest path $P=s\to t$
  can be reported in $O(|P|)$ time.
\end{restatable}
Compared to the data structure of
Roditty and Zwick~\cite{RodittyZ11}
for the unweighted case, our obtained update/query bounds are polynomially higher.
However, our data structure has some very significant advantages.
It is deterministic, handles real-edge-weighted graphs (possibly with negative edge weights and negative cycles),
and the update time bounds holds in the worst case, as opposed to only in the amortized sense~in~\cite{RodittyZ11}.
Moreover, if path reporting is required, the bounds in~\cite{RodittyZ11} hold only against an oblivious adversary.
We also remark that a~slightly more efficient variant of Theorem~\ref{t:shpath}, with $\Ot(mn^{3/4})$ worst-case update
time and $\Ot(n^{3/4})$ query time, can be obtained for the unweighted case.

The near-quadratic space requirement in Theorem~\ref{t:shpath} is clearly undesirable in the sparse setting, but
also applies to all the other known fully dynamic reachability and shortest paths data structures.
Moreover, this phenomenon is not specific to the dynamic setting.
To the best of our knowledge, even for the \emph{static} transitive closure problem, it is not
known whether one can preprocess a general sparse directed graph into a data structure of size $O(n^{2-\eps})$
supporting arbitrary reachability queries in $O(n^{1-\eps})$ time.\footnote{Such a tradeoff is possible, for example,
if the graph has a sublinear \emph{minimum path cover}, see, e.g.,~\cite{MakinenTKPGC19}.}

\subsection{Technical overview}
\paragraph{Shortest paths.} In order to obtain a basic randomized variant of Theorem~\ref{t:shpath}, we combine ideas from the known data structures for fully dynamic APSP
with subcubic worst-case update bound~\cite{AbrahamCK17, GutenbergW20b, Karczmarz21}.
These data structures all build upon hitting set arguments~(dating back to the work of Ullman~and~Yannakakis~\cite{UY91})
yielding a sublinear $\Ot(n/h)$-sized set of vertices of the graph that lie on the shortest paths whose number of edges (\emph{hops}) is at least $h=\poly{n}$.
With this in hand, the main challenge is to recompute pairwise \emph{small-hop} shortest paths, i.e., those with at most $h$ hops, under edge deletions.
As usual, edge insertions are rather easy to handle since the potential new paths created by insertions
necessarily pass through the inserted edges' endpoints.

For efficient recomputation of small-hop paths, our data structure once in a while chooses a collection~$\Pi$ of~$n^2$ pairwise $\leq h$-hop paths in $G$, and a set $C\subseteq V$ of \emph{congested vertices} of truly sublinear (in~$n$) size,
so that the chosen paths are at least as short as shortest $\leq h$-hop paths in $G-C$ (i.e., the graph~$G$ with edges incident to the vertices $C$ removed). The congested vertices are picked in such a way that no individual vertex $v\in V$ appears on the chosen paths too often. As a result, the number of precomputed
paths destroyed by a vertex deletion that have to be restored is bounded. This idea is due to
Probst Gutenberg and Wulff-Nilsen~\cite{GutenbergW20b}.
However, as opposed to~\cite{GutenbergW20b}, we cannot afford to recompute shortest $\leq h$-hop paths upon update in a hierarchical way which is inherently quadratic in $n$
(albeit advantageous in the case of dense graphs).
Instead, recomputation upon deletions is performed using a Dijkstra-like procedure~(as in~\cite{AbrahamCK17}), crucially with the sparsity-aware enhancements of~\cite{Karczmarz21} (such as the degree-weighted congestion scheme).
These techniques, combined with the standard random hitting set argument~\cite{UY91} are enough to get the stated bounds, albeit Monte Carlo randomized.

\paragraph{Derandomization.} Randomization above is only required for the sake of the hitting set argument. Curiously, we do not (and do not know how to) exploit
the often-used property that a random hitting set, once sampled, is valid through multiple versions of the evolving graph
as long as the adversary is oblivious to the hitting set. Therefore, we may as well sample the hitting set from scratch
after each update. This is as opposed to~\cite{AbrahamCK17, Karczmarz21}, where avoiding that
leads to polynomially better bounds.
If a fresh hitting set can be used upon each update, the standard derandomization method
is to use a folklore greedy algorithm (see Lemma~\ref{l:hitting}) for constructing a minimum-size hitting set that is $O(\log{n})$-approximate,
first used in the context of static and dynamic APSP algorithms in~\cite{King99, Zwick02}.
The greedy algorithm runs in linear time in the input size. For constructing a hitting set of explicitly given pairwise
$\leq h$-hop paths, this gives an $O(n^2h)$ time bound per update. This is enough for deterministic variants
of~\cite{AbrahamCK17}~and~\cite{Zwick02}. However, the incurred cost is prohibitive in the sparse case.

Derandomization of our data structure without a polynomial slowdown turns out to be non-trivial and requires
some new tools.
First, when precomputing $\leq h$-hop paths $\Pi$,
we construct a hitting set $H_0$ of those paths in $\Pi$ that have $\Theta(h)$ hops. When $G$ is subject to deletions, $H_0$ hits the precomputed paths
in $\Pi$ that are not destroyed as a result of deletions.
Hence, in order to lift $H_0$ into a hitting set after an update, it is enough to extend it so that it hits all the restored paths.
If we wanted to run the greedy algorithm on the restored paths, the data structure would suffer from a factor-$h$ polynomial slowdown.
This is because the representation of the restored paths (constructed using Dijkstra's algorithm) can be computed more efficiently
that their total hop-length  and encoded using a collection of shortest paths trees $\mathcal{Z}$.
The goal can be thus achieved by finding a hitting set of all $\Theta(h)$-hop root-leaf paths in~$\mathcal{Z}$.
King~\cite{King99}~gave a variant of the aforementioned deterministic greedy algorithm precisely for this task.
The algorithm of~\cite{King99} runs in $O(\min(Nh,|\mathcal{Z}|n))$ time, where $N$ denotes the total size of trees in $\mathcal{Z}$.
While this is optimal when $\mathcal{Z}$ contains $\Theta(n)$ trees of size $\Theta(n)$ (as required in~\cite{King99}),
for small enough $N$ and large enough $|\mathcal{Z}|$, this is not better than the standard greedy algorithm which could
complete the task in $O(Nh)$ time.

We deal with this problem
by designing a novel near-optimal deterministic algorithm computing an $\Ot(n/h)$-sized hitting set of
$h$-hop root-leaf path in a collection of trees
that runs in $O(N\log^2{N})$ time independent of~$h$ (see Theorem~\ref{t:hitting-tree}).
We believe that this algorithm might be of independent interest. The main idea here is to
simulate the greedy algorithm only approximately, which enables taking advantage of dynamic tree data structures~\cite{AlstrupHLT05}.

\subsection{The error in the conference version~\cite{KarczmarzS23}}\label{sec:error}
The initial version of this paper~\cite{KarczmarzS23} (published in ICALP 2023) also claimed a secondary result for fully dynamic reachability. Specifically, in Theorem 2 therein we claimed that for any $t\in[1,\sqrt{m}]$ there exists a Monte Carlo randomized data structure maintaining
  $G$ subject to fully dynamic single-edge updates with $\Ot(mn/t)$ worst-case update time and supporting
  arbitrary-pair reachability queries in $O(t)$ time.
  For sparse graphs, such a data structure would (1) improve upon upon $O(n^{1.529})$/$O(n^{0.529})$ worst-case update/query bounds obtained in~\cite{Sankowski04} and (2) match the amortized update bound of~\cite{RodittyZ08} in the worst case.

Unfortunately, the argument used in~\cite{KarczmarzS23} is flawed. More specifically, the error lies in \cite[Theorem~13]{KarczmarzS23}, where we claim that the inverse of a non-singular sparse matrix $A$ with $m$ non-zero entries over a sufficiently large finite field can be computed in $\Ot(nm)$ time. In the proof sketch provided, it is argued that this follows by the application of Baur-Strassen Theorem~\cite{BaurS83} on top of a variant of the parallel determinant algorithm of Kaltofen and Pan~\cite{KaltofenP91} that implies a randomized algebraic circuit (straight-line program) computing the determinant of $A$ in $\Ot(nm)$ time. Whereas indeed, applying the Baur-Strassen Theorem to the algorithm of~\cite{KaltofenP91} yields an efficient matrix inversion algorithm in the dense case where each matrix element can be treated as the input of the circuit, the sparse variant of the circuit has only $m$ inputs. Therefore, applying the Baur-Strassen theorem can only yield $m$ elements of the inverse $A^{-1}$, when \emph{all $n^2$ elements of $A^{-1}$ are required}.
This deems the proof of the the $\Ot(nm)$ sparse matrix inverse bound~\cite[Theorem~13]{KarczmarzS23} incorrect and consequently invalidates the reachability data structure of~\cite{KarczmarzS23}.
We thank Jan van den Brand (personal communication) for pointing out the shortcoming. 

We note that the flawed data structure in~\cite{KarczmarzS23} can be patched to work for DAGs. Essentially, this follows from a simple fact that in a sparse DAG, all-pairs \emph{path counts} modulo a prime of order $\poly(n)$ can be computed in $O(nm)$ time. For details, see Section~\ref{sec:dag}.

The complexity of sparse matrix inversion over a field remains an interesting problem. To the best
of our knowledge, the current best known bound is $O(n^{2.214})$~\cite{CasacubertaK22} and it holds
even if we only assume that one can apply the input matrix to a vector in a black-box way in $\Ot(n)$ time.

\subsection{Further related work}
Exact all-pairs shortest paths in unweighted graphs have been studied also in partially dynamic settings: incremental~\cite{AusielloIMN91} and decremental~\cite{BaswanaHS07, EvaldFGW21}.
Fully dynamic data structures are also known for $(1+\eps)$-approximate distances in weighted directed graphs~\cite{Bernstein16, BrandN19}.
A significant research effort has been devoted to finding fully- and partially dynamic (approximate) all-pairs shortest paths data structures
for \emph{undirected} graphs, e.g.,~\cite{BernsteinGS21, ChenGHPS20, Chuzhoy21, ChuzhoyZ23, ForsterGNS23, ForsterNP22, abs-2402-18541, BrandFN22}.

Dynamic reachability and shortest paths problems have also been studied from the perspective
of conditional lower bounds~\cite{AbboudW14, GutenbergWW20, HenzingerKNS15, RodittyZ11, BrandNS19}.

\section{Preliminaries}\label{s:prelims}
We work with directed graphs $G=(V,E)$. We denote by $\wei_G(e)\in\mathbb{R}$ the weight of
an edge $uv=e\in E$.
The graph $G$ is called \emph{unweighted} if $\wei_G(e)=1$ for all $e\in E$.
If the graph whose edge we refer to is clear from the context, we may sometimes skip the subscript
and write~$\wei(e)$. For simplicity, we do not allow parallel directed edges between the same endpoints of~$G$,
as those with non-minimum weights can be effectively ignored in reachability and shortest paths problems we study.
As a result, we sometimes write $\wei_G(uv)$ or $\wei(uv)$.

For $u,v\in V$, an $u\to v$ path $P$ in $G$ is formally a sequence of vertices $v_1\ldots v_k\in V$, where $k\geq 1$, $u=v_1$, $v=v_k$,
such that $v_iv_{i+1}\in E$ for all $i=1,\ldots,k-1$.
The hop-length $|P|$ of~$P$ equals $k-1$. The length $\len(P)$ of $P$ is defined
as $\sum_{i=1}^{k-1}\wei_G(v_iv_{i+1})$. $P$ is a \emph{simple path} if $|V(P)|=|E(P)|+1$.
We sometimes view $P$ as a subgraph of $G$ with vertices $\{v_1,\ldots,v_k\}$ and edges (hops) $\{v_1v_2,\ldots,v_{k-1}v_k\}$.

For any $k\geq 0$,
$\dist_G^k(s,t)$ is the minimum length of an $s\to t$ path in $G$ with at most $k$ hops.
A~\emph{shortest $k$-hop-bounded} $s\to t$ path in $G$ is an $s\to t$ path with length $\dist_G^k(s,t)$ and at most~$k$ hops.
We define the $s,t$-distance $\dist_G(s,t)$ as $\inf_{k\geq 0}\dist_G^k(s,t)$.
For $s,t\in V$, we say that~$t$ is \emph{reachable} from $s$ in~$G$ if there exists an $s\to t$ path in $G$,
that is, $\dist_G(s,t)<\infty$.
If $\dist_G(s,t)$ is finite,
there exists a simple $s\to t$ path of length $\dist_G(s,t)$.
Then, we call any $s\to t$ path of length $\dist_G(s,t)$ a \emph{shortest $s,t$-path}.

If $G$ contains no negative cycles, then $\dist_G(s,t)=\dist^{n-1}_G(s,t)$ for all $s,t\in V$.
Moreover, in such a case there exists a \emph{feasible price function} $p:V\to\mathbb{R}$ such
that reduced weight $\wei_p(e)$ satisfies $\wei_p(e):=\wei(e)+p(u)-p(v)\geq 0$ for all $uv=e\in E$.
For any path $s\to t=P\subseteq G$, the reduced length $\len_p(P)$ (i.e., length wrt. weights $\wei_p$) is non-negative and differs
from the original length $\len(P)$ by the value $p(s)-p(t)$ which does not depend on the shape of~$P$.

For any $S\subseteq V$, we denote by $G-S$ the subgraph of $G$ on $V$ obtained from $G$ by removing all edges incident to vertices $S$.

We sometimes talk about rooted out-trees $T$ with all edges directed from a parent to a child.
In such a tree $T$ with root $s$, a \emph{root path} $T[s\to t]$ is the unique path from the root to the vertex $t$ of~$T$.
A subtree of $T$ rooted in some of its vertices $v$ is denoted by $T[v]$.

\section{Fully dynamic shortest paths data structure}
This section is devoted to proving the main theorem of this paper.
\tshpath*
First, let us assume that all the edge weights are non-negative.
Let us also make a simplifying assumption that any shortest $k$-hop-bounded $s\to t$ path
in $G$ always has a minimum possible number of hops and is simple. If there are no negative cycles,
this is easy to guarantee
by replacing each edge weight $\wei(e)$ in $G$ with a pair $(\wei(e),1)$, adding
weights coordinate-wise, and comparing them lexicographically.
We discuss how
to extend the data structure to also handle negative edge weights and negative cycles
later in Section~\ref{s:negative}.

We will first present a simple Monte Carlo randomized data structure,
and show how to make it deterministic with no asymptotic time penalty (wrt. $\Ot(\cdot)$ notation) in Section~\ref{s:deterministic}.

Some further variants of the data structure are sketched in the Appendix. A variant for unweighted
digraphs is given in Section~\ref{s:unweighted}.
In the weighted case, one can also achieve polynomially faster update at the cost of
polynomially slower query and randomization. For details, see~Section~\ref{s:tradeoff}.

The data structure operates in phases of $\Delta$ vertex updates, where $\Delta$ is to be fixed later.
At the beginning of each phase, we apply a rather costly preprocessing described
in the next subsection.

\subsection{Preprocessing at the beginning of a phase}\label{s:prep}

The preprocessing follows the general approach of~\cite{GutenbergW20b} adjusted with some ideas from~\cite{Karczmarz21}.

Let $h\in [2,n]$, and let $\tau$ be a \emph{congestion threshold}, to be set later.
We compute a certain collection of paths $\Pi$ in $G$ containing,
for every pair $s,t\in V$,
at most one $s\to t$ path $\pi_{s,t}$,
satisfying $|\pi_{s,t}|=O(h)$,
and a subset $C\subseteq V$ of \emph{congested vertices}. 

First of all, the collection~$\Pi$ and the set
$C$ satisfy:
\begin{equation}\label{eq:short}
  \dist_G^h(s,t)\leq \len(\pi_{s,t})\leq \dist_{G-C}^h(s,t), \text{ for all }s,t\in V.
\end{equation}
Above, we abuse the notation a bit and set $\len(\pi_{s,t}):=\infty$ if there is no path $\pi_{s,t}$ in $\Pi$.
Moreover, for any $v\in V$, let us define:
\begin{align*}
  \Pi(v)&:=\{\pi_{s,t}\in \Pi:v\in V(\pi_{s,t})\},\\
  \alpha(v)&:=\sum_{\pi_{s,t}\in \Pi(v)} \deg(t).
\end{align*}
Crucially, $\Pi$ additionally satisfies:
\begin{equation}\label{eq:conj}
  \alpha(v) \leq \tau, \text{ for all }v\in V.
\end{equation}

\begin{lemma}\label{l:prep}
  Let $h\in [1,n]$. For any $\tau\geq 2m$, in $O(nmh)$ time one can compute the congested set $C\subseteq V$ and a set of
  paths $\Pi$ satisfying
  conditions~\eqref{eq:short}~and~\eqref{eq:conj} so that $|C|=O(nmh/\tau)$.
\end{lemma}
\begin{proof}
  We start with empty sets $C$ and $\Pi$. Note that~\eqref{eq:conj} is satisfied initially since all values $\alpha(\cdot)$ are zero.
    We will gradually add new paths to $\Pi$ while maintaining~\eqref{eq:conj} and ensuring
    that~\eqref{eq:short} holds for more and more pairs $s,t$. While introducing
    new paths to $\Pi$, we will also maintain the values $\alpha(v)$ (as defined above) for all $v\in V$.

    We process source vertices $s\in V$ one by one, in arbitrary order. For each such $s$,
    we first move to $C$ all the vertices $v\in V\setminus C$ with $\alpha(v)> \tau/2$.
    Next, we compute, for all $t\in V$,
    a shortest $h$-hop-bounded path $\pi_{s,t}=s\to t$ in $G-C$ (if such a path exists).
    For a fixed $s$, all the paths $\pi_{s,t}$ can be computed in $O(mh)$ time using a variant of Bellman-Ford algorithm.
    We add the newly computed paths to $\Pi$.
    Note that afterwards,~\eqref{eq:short} holds for $s$ and all $t\in V$, even if the respective $\pi_{s,t}$ does
    not exist. Moreover,~\eqref{eq:short}
    also holds for all $\pi_{s',t'}\in\Pi$ that have been added for a source $s'$ processed earlier than $s$.
    Indeed, extending the set $C$ only weakens the upper bound in~\eqref{eq:short}.
    The values $\alpha(v)$ can be updated easily in $O(nh)$ time.
    Observe that for any $v\in V\setminus C$, $\alpha(v)$ grows by at most
    $\sum_{t\in V}\deg(t)=m$ when processing $s$. As a result, after processing~$s$, we have $\alpha(v)\leq \tau/2+m\leq \tau/2+\tau/2=\tau$
    and hence~\eqref{eq:conj} is satisfied.
    At the same time, since we use paths from $G-C$, for any $y\in C$, $\alpha(y)$ does not increase and thus we still have $\alpha(y)\leq \tau$.

    Finally, note that for any $\pi_{s,t}$ added to $\Pi$, since $|\pi_{s,t}|\leq h$, $\alpha(v)$
    grows by $\deg(t)$ for at most~$h$ distinct vertices $v$.
    As a result, we have $\sum_{v\in V}\alpha(v)\leq \sum_{t\in V}\deg(t)\cdot \left(\sum_{s\in V}h\right)\leq m\cdot nh$.
    But for each $y\in C$, we have $\alpha(y)>\tau/2$, so there is at most $2nmh/\tau$ such vertices $y$.
\end{proof}

Applying Lemma~\ref{l:prep} constitutes the only preprocessing that we perform at the beginning of a phase in the Monte Carlo randomized variant.
The computed paths $\Pi$ are stored explicitly and thus the used space might be $\Theta(n^2h)$.
Note that with the help of additional $O\left(\sum_{\pi_{s,t}\in\Pi}|\pi_{s,t}|\right)=O(n^2h)$ vertex-path pointers, we can
report the elements of any $\Pi(v)$, $v\in V$, in constant time per element.
We will discuss how to improve the space to $\Ot(n^2)$ using a trick due to Probst Gutenberg and Wulff-Nilsen~\cite{GutenbergW20b} in Section~\ref{s:space}.
\subsection{Update}
When a phase proceeds, let $D$ be the set of at most $\Delta$ \emph{affected} vertices in the current phase, that is,~$D$
contains every $v$ such that a vertex update around $v$ has been issued in this phase.

In the query procedure, we will separately consider paths going through $C\cup D$, and those
lying entirely in $G-(C\cup D)$. To handle the former, upon each update we simply compute single-source shortest-path
trees from and to each $s\in C\cup D$ in the current graph $G$. This takes $\Ot(|C\cup D|m)$ worst-case time
using Dijkstra's algorithm.

As a matter of fact, we will not quite compute shortest paths in $G-(C\cup D)$, but instead, we will find paths in $G-D$
that are not longer than the distances between their corresponding endpoints in $G-(C\cup D)$.
This is acceptable since $G-D\subseteq G$.

To prepare for queries about the paths in $G-(C\cup D)$, we do the following.
We will separately handle \emph{short} $\leq h$-hop shortest paths, and \emph{long} $>h$-hop shortest paths.

\paragraph{Short paths.} Denote by $G_0$ the graph at the beginning of the phase. Recall that we use $G$ to refer
to the current graph. Clearly, we have $G-D\subseteq G_0$.
Fix some $s\in V$.
First of all, note that if for some $t\in V$, $V(\pi_{s,t})\cap D=\emptyset$, then
$\pi_{s,t}\subseteq G-D$, so by~\eqref{eq:short}:
\begin{equation*}
  \dist_{G-D}(s,t)\leq \len(\pi_{s,t})\leq \dist^h_{G_0-C}(s,t)\leq \dist^h_{G-(C\cup D)}(s,t).
\end{equation*}
The paths $\pi_{s,t}$ going through $D$ are not preserved in $G-(C\cup D)$ and thus we cannot use them.
We replace them with other paths $\pi'_{s,t}$ constructed using the following lemma.
\begin{lemma}\label{l:dijkstra-rebuild}
  For $s\in V$, let $Q_s$ contain all $t$ such that $V(\pi_{s,t})\cap D\neq\emptyset$.
  In $\Ot\left(\sum_{t\in Q_s}\deg(t)\right)$ time we can compute a representation
  of paths $\pi'_{s,t}\subseteq G-D$ (where $t\in Q_s$), each with possibly $\Theta(n)$ hops, satisfying:
\begin{equation*}
  \dist_{G-D}(s,t)\leq \len(\pi'_{s,t}) \leq \dist^h_{G-(C\cup D)}(s,t).
\end{equation*}
  The representation is a tree $T_s$ rooted at $s$ such that:
  \begin{enumerate}[(1)]
    \item some edges $sv\in E(T_s)$ represent paths $\pi_{s,v}\subseteq G-D$ from $\Pi$ and have corresponding weights $\len(\pi_{s,v})$,
    \item all other edges of $T_s$ come from $E(G-D)$,
    \item for all $t\in Q_s$, $\pi'_{s,t}$ equals $T_s[s\to t]$ with possibly the first edge $sw$ of that path uncompressed into the corresponding path $\pi_{s,w}\in\Pi$.
  \end{enumerate}
  \end{lemma}
\begin{proof}
  Let $Y$ be an edge-induced directed graph obtained as follows. For all $t\in Q_s$, and every of
  at most $\deg(t)$ edges $vt\in E(G-D)$, we add to $Y$ the following:
  \begin{itemize}
    \item the edge $vt$ itself (with the same weight),
    \item if $V(\pi_{s,v})\cap D= \emptyset$, an edge $sv$ of weight $\len(\pi_{s,v})$
      corresponding to the path $\pi_{s,v}\in \Pi$.
  \end{itemize}
  
  The algorithm is to simply compute a shortest paths tree $T_s$ from $s$ in $Y$ in \linebreak
  $\Ot(|E(Y)|)=\Ot\left(\sum_{t\in Q_s}\deg(t)\right)$ time using Dijkstra's algorithm. Clearly,
  any path $T_s[s\to t]$ corresponds to an $s\to t$ path in $G-D$.
  It is thus sufficient
  to prove that for all $t\in Q_s$, we have $\len(T_s[s\to t])\leq \dist^h_{G-(C\cup D)}(s,t)$.

  If $t$ is unreachable in $G-(C\cup D)$ from $s$ using a path
  with at most $h$ hops, there is nothing to prove.
  Otherwise,
  let a simple path $P$ be a shortest $h$-hop-bounded $s\to t$ path in $G-(C\cup D)$. Let $p$ be the last vertex on $P$ such that $V(\pi_{s,p})\cap D=\emptyset$, that is, $p\notin Q_s$.
  Note that $p$ exists since $\dist_{G-C}^h(s,v)\neq\infty$ for all $v\in V(P)$ (which implies $\pi_{s,v}\in\Pi$) and $p\neq t$.
  Let~$P'$ be the $s\to p$ subpath of $P$.
  Let $e_1,\ldots,e_k\in E(G-(C\cup D))$ be the edges following~$p$ on~$P$.
  Here, $p$ is the tail of $e_1$.
  By the definition of $Y$ and $p$, we have $e_i\in E(Y)$ for all $i=1,\ldots,k$
  since the head of each $e_i$ is in $Q_s$. 
  Moreover, there is an edge $sp$ of weight $\len(\pi_{s,p})$ in $Y$.
  Now, since $\pi_{s,p}$ is a path
  in $G-D$ of length at most $\dist_{G_0-C}^h(s,p)$, whereas the path $P'\subseteq G-(C\cup D)=G_0-(C\cup D)$ has less than $h$ hops, we obtain
  $\len(\pi_{s,p})\leq \len(P')$ and hence:
  \begin{equation*}
    \len(T_s[s\to t])=\dist_Y(s,t)\leq \len(\pi_{s,p})+\sum_{i=1}^{k}\wei(e_i)\leq \len(P')+\sum_{i=1}^k w(e_i)=\len(P)=\dist^h_{G-(C\cup D)}(s,t).\qedhere
  \end{equation*}
\end{proof}

We compute the paths $\pi'_{s,t}$ from Lemma~\ref{l:dijkstra-rebuild}
for all $s\in V$, $t\in Q_s$. Recall that $t\in Q_s$ implies that $V(\pi_{s,t})\cap D\neq \emptyset$ and thus $\pi_{s,t}\in\Pi(d)$ for some $d\in D$.
Therefore, the time needed for computing the paths $\pi'_{s,t}$ can be bounded as follows:
\begin{equation*}
  \Ot\left(\sum_{s\in V}\sum_{t\in Q_s}\deg(t)\right)
  = \Ot\left(\sum_{d\in D}\sum_{\pi_{s,t}\in \Pi(d)}\deg(t)\right)=\Ot\left(\sum_{d\in D}\alpha(d)\right)=\Ot(|D|\tau)=\Ot(\Delta\tau).
\end{equation*}
Note that the sets $Q_s$ can also be constructed within this bound: they can be read
from $\bigcup_{d\in D}\Pi(d)$ which also has size $\Ot(\Delta\tau)$ and the paths from any $\Pi(v)$ can be reported in $O(1)$ time per path.

For all $s\in V$ and $t\notin Q_s$, let us simply set $\pi'_{s,t}:=\pi_{s,t}$. Let us also put $\Pi'=\{\pi_{s,t}':s,t\in V\}$.
To summarize, in $\Ot(\Delta\tau)$ time we can find, for all $s,t\in V$,
a representation of paths $\pi_{s,t}'$ in $G-D$ that are at least as short as the corresponding
shortest $h$-hop-bounded $s\to t$ paths in $G-(C\cup D)$. Storing a representation of the paths $\Pi'\setminus \Pi$
requires $\Ot(\min(\Delta\tau,n^2))$ additional space since, by the construction of Lemma~\ref{l:dijkstra-rebuild},
each of these paths can be encoded using its last edge and a pointer to another path in $\Pi'$
with less hops.

\paragraph{Long paths.}
In order to handle long paths, we use the following standard hitting set trick from~\cite{UY91}.
\begin{lemma}\label{l:hitting-rand}
  Let $G=(V,E)$ be a directed graph with no negative cycles. For any $s,t\in V$, fix some simple shortest $s\to t$ path $p_{s,t}$ in $G$.

  Let $H\subseteq V$ be obtained by sampling, uniformly and independently (also from the choice of paths $p_{s,t}$), $c\cdot (n/h)\log{n}$ elements
  of $V$, where $c\geq 1$ is a constant. Then, with high probability\footnote{That is, with probability at least $1-1/n^\alpha$, where
    the constant $\alpha\geq 1$ can be set arbitrarily. We will also use the standard abbreviation w.h.p.}
 (controlled by the constant $c$), for all $s,t\in V$,
  if $|p_{s,t}|\geq h$, then $V(p_{s,t})\cap H\neq\emptyset$.
\end{lemma}
On update, we simply apply Lemma~\ref{l:hitting-rand} to the graph $G-(C\cup D)$ and an arbitrary choice of pairwise
shortest paths therein. This way, with high probability,
we obtain an $\Ot(n/h)$-sized hitting set $H$ of shortest paths in $G-(C\cup D)$ that have at least $h$ hops.
Finally, we simply compute shortest paths trees from and to the vertices $H$ in $G-(C\cup D)$
in $\Ot(|H|m)=\Ot(mn/h)$ worst-case time using Dijkstra's algorithm.

\subsection{Query}
To answer a query about $s,t$ distance in the current graph, we simply return:
\begin{equation}\label{eq:query}
  \min\left(\min_{v\in C\cup D}\left\{\dist_G(s,v)+\dist_G(v,t)\right\},\min_{v\in H}\left\{\dist_{G-(C\cup D)}(s,v)+\dist_{G-(C\cup D)}(v,t)\right\},\len(\pi'_{s,t})\right).
\end{equation}
The first term above is responsible for considering all $s,t$ paths in $G$ going through $C\cup D$.
If no shortest $s,t$ paths in $G$ pass through $C\cup D$, then the second term
captures (with high probability) one of such paths provided that it has at least~$h$ hops.
Finally, if every shortest $s,t$ path in $G$ does not go through $C\cup D$ and has
less than $h$ hops, then $\dist_{G}(s,t)=\dist_{G-D}(s,t)=\dist^h_{G-(C\cup D)}(s,t)$.
Moreover, by Lemma~\ref{l:dijkstra-rebuild}, a path
$\pi'_{s,t}$ is contained in $G-D\subseteq G$ and we have
\begin{equation*}
  \dist_G(s,t)=\dist_{G-D}(s,t)\leq \len(\pi'_{s,t})\leq \dist_{G-(C\cup D)}^h(s,t)=\dist_G(s,t),
\end{equation*}
so indeed $\dist_G(s,t)=\len(\pi'_{s,t})$.

Finally, note that finding the minimizer in~\eqref{eq:query} allows for reconstruction of
some shortest $s,t$ path $P$ in $G$ in $O(|P|)$ time using the stored data structures.

\subsection{Time analysis}
The total time spent handling a single update is:
\begin{equation*}
  \Ot\left((|D|+|C|+|H|)m+\Delta\tau\right)=\Ot(m\Delta+nm^2h/\tau+mn/h+\Delta\tau).
\end{equation*}
There is also an $O(mnh)$ preprocessing cost spent every $\Delta$ updates which yields
an amortized cost of $\Ot(mnh/\Delta)$ per update.
Since $\tau\geq 2m$, the term $m\Delta$ is negligible above.

Balancing the terms $mnh/\Delta$ and $mn/h$ yields $\Delta=h^2$.
Next, balancing with $\Delta\tau$ yields ${\tau=mn/h^3}$ under the assumption $h=O(n^{1/3})$.
Finally, balancing $mn/h$ and $nm^2h/\tau=mh^4$ yields $h=n^{1/5}$, $\Delta=n^{2/5}$, and $\tau=mn^{2/5}$.
For such a choice of parameters, the amortized update time is $\Ot(mn^{4/5})$.
Since the only source of amortization here is a costly preprocessing step happening
in a coordinated way every $\Delta$ updates, the bounds can be made worst-case
using a standard technique, see, e.g.,~\cite{AbrahamCK17, BrandNS19}.

The query time is $O(|C|+|D|+|H|+1)$. For the obtained parameters, the bound becomes \linebreak
$\Ot(\Delta+nmh/\tau+n/h)=\Ot(h^4+n/h)=\Ot(n^{4/5})$.

\begin{remark}
  In the above analysis, we have silently assumed that the ``current'' number of edges $m$ does not decrease significantly (say, by more than a constant factor)
  during a phase due to vertex deletions, so that $m=\Omega(m_0)$ holds at all times, where $m_0=|E(G_0)|$. 
  Since the preprocessing of Lemma~\ref{l:prep} is applied to $G_0$, for the chosen parameters $h,\Delta$, and $\tau=m_0n^{2/5}$,
  the update bound should more precisely be bounded by $\Ot(\max(m,m_0)\cdot n^{4/5})$.
  In general, it might happen that $m$ becomes polynomially smaller that $m_0$ while the phase proceeds, e.g., if $m_0=O(n\Delta)$.
  This could make the update bound higher than $\Ot(mn^{4/5})$.

  There is a simple fix to this shortcoming, described in~{\upshape\cite{Karczmarz21}}:
  when a phase starts, it is enough to put aside a set $B\subseteq V$ of $\Delta$ vertices
  with largest degrees in $G_0$ and preprocess the graph $G_0-B$ instead. 
  The edges incident to vertices $B$ can be viewed as added during the first $\Delta$
  ``auxiliary'' updates in the phase, and effectively included in the affected set $D$ from the beginning
  of the phase. One can easily prove that
  this guarantees that $m=\Omega(m_0)$ throughout the phase, where $m_0$ is now defined as $|E(G_0-B)|$.
\end{remark}

\subsection{Derandomization}\label{s:deterministic}
The only source of randomization so far was sampling a subset of vertices that hits
shortest paths in $G-(C\cup D)$ with at least $h$ hops.
To derandomize the data structure, we will construct a hitting set $H$ of size $\Ot(n/h)$
such that $H$ hits all the paths in $\Pi'=\{\pi'_{s,t}:s,t\in V\}$ (constructed during update) with at least~$h$ distinct vertices.
Recall that the paths $\Pi'$ have been used to handle short paths so far.
We first show that a hitting set $H$ defined this way can serve the same purpose as the randomly sampled hitting set.

\begin{lemma}\label{l:hitting-correct}
  Let $H\subseteq V$ be such that for all $s,t\in V$ satisfying  $|V(\pi'_{s,t})|\geq h$,
  $V(\pi'_{s,t})\cap H\neq \emptyset$ holds. Let $a,b\in V$ be such that every shortest $a\to b$ path
  in $G$ has more than $h$ hops and \emph{does not} go through $C\cup D$. Then there exists a shortest $a\to b$ path in $G$ that goes through a vertex of~$H$.
\end{lemma}
\begin{proof}
  Let $Q$ be the shortest $a\to b$ path in $G$ that has the minimum number of hops. By the assumption, $|Q|>h$
  and $V(Q)\cap (C\cup D)=\emptyset$.
  Let $Q=RS$, where $R=a\to c$ is the $h$-hop prefix of $Q$.
  We have $R\subseteq G-(C\cup D)$ and, since $Q$ is a shortest path in $G$,
  $R$ is also shortest in $G$ and
  \begin{equation*}
  \len(R)=\dist_{G}(a,c)=\dist_{G}^h(a,c)=\dist_{G-(C\cup D)}^h(a,c).
  \end{equation*}
  Since $\dist_{G-(C\cup D)}^h(a,c)$ is finite, the path $\pi'_{a,c}\subseteq G-D\subseteq G$ satisfies
  \begin{equation*}
    \dist_G(a,c)\leq \dist_{G-D}(a,c)\leq \len(\pi'_{a,c})\leq \dist_{G-(C\cup D)}^h(a,c)=\dist_G(a,c).
  \end{equation*}
  We conclude that the path $Q'=\pi'_{a,c}\cdot S$ satisfies $\len(Q')=\len(Q)$ and thus $Q'$ is also a shortest $a\to b$ path in $G$.
  Since $G$ has no negative cycles, one can obtain a \emph{simple} $a\to c$ path $\pi''_{a,c}$ from $\pi'_{a,c}$ by eliminating zero-weight cycles,
  so that $\len(\pi''_{a,c})=\len(\pi'_{a,c})=\dist_G(a,c)$ and $V(\pi''_{a,c})\subseteq V(\pi'_{a,c})$.
  By the definition of $Q$, $|V(\pi'_{a,c})|\geq |\pi''_{a,c}|\geq |R|\geq h$, since otherwise $Q$ would not have a minimum number of hops.
  As a result, from the assumption we conclude $V(\pi'_{a,c})\cap H\neq \emptyset$, so~$Q'$ is a shortest $a\to b$ path in $G$ going through a vertex of $H$.
\end{proof}

\paragraph{Additional preprocessing.} When a phase starts, we additionally do the following.
Let $\Pi_0$ be a set of paths obtained as follows. For all $\pi_{s,t}\in \Pi$, if $|\pi_{s,t}|\geq h/2$, we add $\pi_{s,t}$ to $\Pi_0$.

Let us now recall a folklore greedy algorithm (used, e.g.,~in~\cite{Zwick02}) for computing a hitting set of a collection of
sufficiently large sets over a common ground set, summarized by the following lemma.

\begin{lemma}\label{l:hitting}
  Let $X$ be a ground set of size $n$ and let $\mathcal{Y}$ be a family of subsets of $X$, each with at least $k$ elements.
  Then, in $O\left(\sum_{Y\in\mathcal{Y}}|Y|\right)$ time one can deterministically compute a \emph{hitting set} $H\subseteq X$
  of size $O((n/k) \cdot\log{n})$ such that $H\cap Y\neq\emptyset$ for all $Y\in \mathcal{Y}$.
\end{lemma}
We skip the proof of Lemma~\ref{l:hitting} since we later prove a more general result in~Theorem~\ref{t:hitting-tree}.
Recall that, by our simplifying assumptions, all the paths in $\Pi$, and thus in $\Pi_0$, are simple.
Using Lemma~\ref{l:hitting} we can compute a hitting set $H_0\subseteq V$ of $\Pi_0$ in $O(n^2h)$ time. $H_0$ has size $\Ot(n/h)$.

\paragraph{Computing a hitting set upon update.} To compute a hitting set $H\subseteq V\setminus D$
as required by Lemma~\ref{l:hitting-correct}, we perform the following additional steps upon update.
Recall that the precomputed set $H_0\subseteq V$ hits all (simple) paths in $\Pi\cap \Pi'$
with at least $h/2$ hops, and thus also those in $\Pi\cap \Pi'$ that have at least~$h$ distinct vertices.
We will augment $H_0$ into $H$ so that it also hits all the paths in $\Pi'\setminus \Pi$
with at least $h$ distinct vertices.

Recall from Lemma~\ref{l:dijkstra-rebuild} that for a fixed $s\in V$, all the paths $\pi'_{s,t}$,
where $t\in Q_s$,
are encoded using a tree $T_s$.
By construction, for each edge $e$ of $T_s$, we have that the tail of $e$ is $s$,
or the head of $e$ is in~$Q_s$.
Consider a subtree $T_s[u]$ rooted at some child $u$ of $s$ in $T_s$.
If the edge $su$ in $T_s$ corresponds to the path $\pi_{s,u}$ with $|\pi_{s,u}|\geq h/2$
then $H_0$ hits $\pi_{s,u}$.
As a result, for all $t\in Q_s\cap V(T_s[u])$, $V(\pi_{s,u})\subseteq V(\pi'_{s,t})$ and hence
if $|V(\pi'_{s,t})|\geq h$ then $H_0$ hits $V(\pi'_{s,t})$.
Otherwise either $su$ is a single edge from $G-D$, or it corresponds to a path $\pi_{s,u}$ with $|\pi_{s,u}|<h/2$.
In either of these cases, if some $t$ is at depth less than $h/2-1$
in $T_s[u]$, then $|V(\pi'_{s,t})|<h/2+1+(h/2-1)=h$, so the path $\pi'_{s,t}$ does not need to be hit by $H$.
Consequently, observe that it is enough for $H$ to hit all the $(h/2-1)$-hop root paths in $T_s[u]$
in order to have $V(\pi'_{s,t})\cap H\neq \emptyset$ for each $t\in T_s[u]$ with $|V(\pi'_{s,t})|\geq h$.

Let $\mathcal{Z}$ be the collection of all the subtrees $T_s[u]$, where $s\in V$ and $su\in E(T_s)$.
It is now enough to compute an $\Ot(n/h)$-sized hitting set $H_1$ of each of the $(h/2-1)$-hop root paths in all trees in~$\mathcal{Z}$.
Then, $H_0\cup H_1$ will form a desired hitting set $H$ of all the paths in $\Pi'$ with at least $h$ distinct vertices.
To this end, we could use a well-known variant of Lemma~\ref{l:hitting} due to King~\cite[Lemma~5.2]{King99}.
However, the running time of that algorithm cannot be easily bounded with the total size $N$
of $\mathcal{Z}$ (i.e., $N=\sum_{T\in\mathcal{Z}}|T|$) exclusively;
its running time is $O\left(N+\sum_{T\in \mathcal{Z}}\min(n\log{n},|T|k)\right)=O(\min(Nk,|\mathcal{Z}|n\log{n}))$
if one desires to hit $k$-hop root paths.
Though, for some important cases, e.g., when $\mathcal{Z}$ contains $n$ trees with $\Theta(n)$ vertices each, the
running time is near-linear in $N$ for any~$k$.
Unfortunately, this might not be the case in our scenario.
Instead, we present a more sophisticated near-linear (independent of~$k$) time algorithm for this task.

\begin{theorem}\label{t:hitting-tree}
  Let $V$ be a vertex set of size $n$ and let $\mathcal{Z}$ be a family of trees on~$V$.
  Let ${N=\sum_{T\in \mathcal{Z}}|T|}$.
  For any $k\in [1,n]$, in $O(N\log^2{n})$ time one can deterministically compute an $O(n/k \cdot\log{n})$-sized \emph{hitting set} $H\subseteq V$
  of all the $k$-hop root paths in all the trees in $\mathcal{Z}$.
\end{theorem}
\begin{proof}
  We first iteratively prune the trees in $\mathcal{Z}$ of all the leaves at depths not equal to $k$:
  this does not alter the set of subpaths required to be hit.
  Afterwards, the task is to hit all the root-leaf paths in the collection~$\mathcal{Z}$, each of exactly
  $k$ hops.

  Similarly as in~\cite{King99}, we would like to simulate the greedy algorithm behind Lemma~\ref{l:hitting}, that is,
  repeatedly pick a vertex $v\in V$ hitting the largest number of paths not yet hit,
  and add it to the constructed set $H$.
  However, we cannot afford to follow this approach directly. Instead, when $L\geq 1$ paths
  are remaining to be hit, and there is $n'\geq k+1$ vertices $V\setminus H$ that have not yet been chosen, we 
  pick a vertex hitting at least $\frac{L(k+1)}{2n'}$ remaining paths.
  Note that there always exists a vertex hitting at least $\frac{L(k+1)}{n'}$ remaining paths, since otherwise some
  of the remaining paths would contain a vertex from outside $V\setminus H$, a contradiction.
  A~single step in our approach reduces $L$ to at most $\left(1-\frac{k+1}{2n'}\right)L$, so $\lceil 2n'/(k+1)\rceil=O(n/k)$ steps
  reduce $L$ to at most $L/e$.
  Hence, since $L$ is an integer, after $O\left(\frac{n}{k}\ln{L}\right)=O\left(\frac{n}{k}\ln{N}\right)$ steps
  $L$ will drop to $0$, i.e., all required paths will be hit.

  Our strategy can be also rephrased as follows: maintain $2$-approximate counters \linebreak ${\{c_v:v\in V\}}$
  such that the vertex $v$ hits between $c_v$ and $2c_v$ of the remaining paths, and repeatedly
  pick a vertex $z$ with the maximum value of $c_z$. By the above discussion, the picked $z$
  will always satisfy $c_z\geq L(k+1)/2n'$, as desired.
  To implement this strategy, we proceed as follows.

  For each $T\in \mathcal{Z}$, and $v\in V(T)$, let $d_{v,T}$ be the \emph{exact}
  number of \emph{previously not hit} root-leaf paths in $T$ that $v$ hits.
  Note that through the entire collection, $v$ hits $D_v:=\sum_{T\in\mathcal{Z}} d_{v,T}$ paths not yet hit.
  Observe that when a root-leaf path to the leaf $l$ in $T$ is hit for the first time, the value
  $d_{v,T}$ of all the ancestors $v$ of $l$ gets decreased by one.
  In fact, the algorithm of~\cite{King99} can be seen to maintain such values $d_{v,T}$ and $D_v$ explicitly.
  However, this is too costly for us; we will instead maintain the exact values $d_{v,T}$ only implicitly, in a data structure.

  For each $T\in\mathcal{Z}$, we keep $V(T)$ (explicitly) partitioned into subsets
  $V_{T,0},\ldots,V_{T,\ell}$, where ${\ell=O(\log{|V(T)|})}$,
  so that $v\in V_{T,i}$ iff $d_{v,T}\in [2^i,2^{i+1})$.
  Throughout the process, the values $d_{v,T}$ will only decrease, so
  a vertex $v\in V(T)$ can only move $O(\log{n})$ times
  to a subset $V_{T,j}$ with a lower value~$j$.
  Let us first argue that maintaining such partitions yields
  the desired 2-approximate counters rather straightforwardly.
  
  For $v\in V$, let us define $c_v=\sum_{T,j:v\in V_{T,j}} 2^j$. Then, we have:
  \begin{equation*}
    D_v=\sum_{T\in\mathcal{Z}}d_{v,T}\geq \sum_{T,j:v\in V_{T,j}} 2^j=c_v=\sum_{T,j:v\in V_{T,j}} 2^j=\frac{1}{2}\sum_{T,j:v\in V_{T,j}} 2^{j+1}> \frac{1}{2}\sum_{T\in\mathcal{Z}}d_{v,T}=\frac{1}{2}D_v.
  \end{equation*}
  As a result, the counters $c_v$ indeed $2$-approximate the values $D_v$ and can be maintained~subject to changes
  in the partitions $V_{T,i}$, for all $T,i$, in $O\left(\sum_{T\in\mathcal{Z}}|T|\log{n}\right)=O(N\log{n})$ time.

  Fix some $T\in\mathcal{Z}$. To maintain the partition $V_{T,0},\ldots,V_{T,\ell}$, we maintain
  the values $d_{v,T}$ using~$\ell$ data structures $S_{T,0},\ldots,S_{T,\ell}$.
  The data structure $S_{T,i}$ associates (implicitly) the following vertex weights to the individual vertices $v$ of $T$.
  If $d_{v,T}\geq 2^i$, then $v$ has weight $d_{v,T}$ in $S_{T,i}$.
  Otherwise, if $d_{v,T}<2^i$, then $v$ has weight $\infty$ in $S_{T,i}$.
  In particular, $S_{T,0}$ associates the exact values $d_{v,T}$ to the vertices of~$T$.

  Fix some $i=0,\ldots,\ell$. $S_{T,i}$ is implemented using, e.g., a top-tree~\cite[Theorem~2.4]{AlstrupHLT05} that allows performing the following operations, both in $O(\log{n})$ time\footnote{As a matter of fact, in~\cite{AlstrupHLT05} this is shown for edge weights. However, vertex weights can be simulated easily using edge weights by assigning each vertex its parent edge, and explicitly maintaining the weight of the root.}:
  \begin{enumerate}[(1)]
    \item adding the same $\delta\in\mathbb{R}$ to the weights of vertices on some specified path in the tree, and
    \item querying for a vertex of the tree with minimum weight.
  \end{enumerate}
  Clearly, $S_{T,i}$ can
  be initialized at the beginning of the process in $O(|T|\log{n})$ time.
  When a new vertex $z$ is added to $H$, and $z\in V(T)$, we iterate through all the
  (previously unvisited) descendants of $z$ to identify the (original) leaves~$y$ at depth $k$ such that
  the root-to-$y$ path in~$T$ has not been previously hit.
  For each such $y$, we decrease the weights of all the ancestors of $y$ in $T$ (all lying
  on a single path in $T$) by $1$. This requires a single top-tree operation on~$S_{T,i}$.
  Afterwards, for all $w\in V(T)$ whose value $d_{w,T}$ was at least $2^i$ before adding $z$
  to~$H$, $S_{T,i}$ contains (in an implicit way) the correctly updated exact value $d_{w,T}$.
  Some of these values in $S_{T,i}$ might drop below $2^i$, though.
  To deal with this, we repeatedly attempt to extract the minimum-valued vertex $x\in V(T)$
  from $S_{T,i}$. If the value of $x$ is less than $2^i$, we reset the value of $x$ in $S_{T,i}$
  to~$\infty$. Otherwise, we stop; at this point all the values in $S_{T,i}$ are at least $2^i$; the invariant posed on $S_{T,i}$ is fixed.

  The above update procedure is performed for each $i$. Observe that $v\in V_{T,i}$
  iff $i$ is the maximum index such that $v$ has assigned a finite value in $S_{T,i}$.
  Since for all $i$ we can explicitly track which vertices in $S_{T,i}$ are assigned $\infty$ while performing updates,
  the time needed to maintain the partition $V_{T,0},\ldots,V_{T,\ell}$
  can be charged to the cost of maintaining the data structures $S_{T,0},\ldots,S_{T,\ell}$.

  Let us now analyze the time cost of this algorithm.
  For each $T\in\mathcal{Z}$, we iterate through every vertex of $T$ at most $O(1)$ times. For $i=0,\ldots,\ell$,
  at most $O(|V(T)|+|H\cap V(T)|)=O(|V(T)|)$ top-tree operations are performed on $S_{T,i}$.
  Hence, the cost of maintaining all $S_{T,i}$ for all ${i=0,\ldots,O(\log{n})}$ is $O(|T|\log^2{n})$.
  Through all $T\in\mathcal{Z}$, this is $O(N\log^2{n})$.
  
  To implement finding a next vertex $z\in H$ with the largest~$c_z$, one may simply store the counters~$c_z$ in a priority queue.
  Since the counters are updated $O(N\log{n})$ times in total, the priority queue operations cost is $O(N\log^2{n})$ as well.
\end{proof}

Observe that through all $s$, the total number of edges in trees added to $\mathcal{Z}$
can be bounded by the number of edges in the (compressed) trees $T_s$ of Lemma~\ref{l:dijkstra-rebuild},
and thus also by $\Ot(\min(\tau\Delta,n^2))$.
As a result, by Theorem~\ref{t:hitting-tree}, the desired set $H$ hitting all paths $\pi'_{s,t}$
with at least $h$ distinct vertices can be computed in $\Ot(\tau\Delta)$ time, using at most quadratic space.
This does not increase the running time of the update procedure in the asymptotic sense.

\subsection{Reducing the space usage}\label{s:space}
So far, the space used by the preprocessing phase could only be bounded by $O(n^2h)$
as we have explicitly stored the $O(n^2)$ preprocessed paths
$\pi_{s,t}\in\Pi$, each with $O(h)$ hops.

We do not, however, need to store the paths $\pi_{s,t}\in \Pi$ explicitly. For performing updates and
answering distance queries,
we only require
storing the values $\len(\pi_{s,t})$, $|\pi_{s,t}|$, and being able to efficiently access the sets $\Pi(v)$,
for any $v\in V$. If we want to also support path queries, then
constant-time reporting of the subsequent edges of $\pi_{s,t}$ is also needed.
Probst Gutenberg and Wulff-Nilsen~\cite[Section~4.2]{GutenbergW20b} showed an elegant way of achieving that in a slightly relaxed way using only
$\Ot(n^2\log{h})$ space.
\begin{lemma}\label{l:space}{\upshape\cite{GutenbergW20b}}
  Let $G=(V,E)$ be a real-weighted digraph with no negative cycles. Let $s\in V$ and let $h\in [1,n]$. Using $O(mh)$ time
  and $O(nh)$ space, one can build an $\Ot(n)$-space data structure representing
  a collection $\{\pi_t:t\in V\}$ of (not necessarily simple) $O(h)$-hop paths from $s$ to all other vertices in $G$ such that for any $t$,
  $\len(\pi_t)\leq \dist^h_G(s,t)$.

  For any $v\in V$, the data structure allows:
  \begin{itemize}
    \item accessing $\len(\pi_v)$ and $|\pi_v|$ in $O(1)$ time,
    \item reporting the set $P_v=\{t\in V:v\in V(\pi_t)\}$ in $\Ot(|P_v|)$ time,
    \item reporting the edges of $\pi_v$ in $O(|\pi_v|)$ time.
  \end{itemize}
\end{lemma}
\begin{proof}[Proof sketch]
  Suppose we compute shortest $h$-hop-bounded $s\to t$ paths $p_t$ from $s$ to all $t\in V$.
  This takes $O(mh)$ time, but storing the computed paths explicitly would require $\Theta(nh)$ space.
  Recall that if $G$ has no negative cycles, then we may wlog. assume that the paths $p_t$ 
  are all simple.
  As a result, one can deterministically compute an $\Ot(n/h)$-sized hitting set~$H$
  of the $\lceil h/3\rceil$-hop infixes starting at the $(\lceil h/3\rceil+2)$-th hop
  of those of the computed $p_t$ that satisfy $|p_t|\geq \lceil 2h/3\rceil$.
  We explicitly store the paths $p_y$ for all $y\in H$ which costs only $\Ot(|H|\cdot h)=\Ot(n)$ space.

  Let $G'$ be obtained from $G$ be adding shortcut edges $e_y=sy$ of weight $\wei(e_y)=\len(p_y)$
  for all $y\in H$.
  Note that for all $v\in V$, $\dist_{G'}^{\lceil 2h/3\rceil}(s,v)\leq \dist_{G}^h(s,v)=\len(p_v)$
  and every $\leq \lceil 2h/3\rceil$-hop path in $G'$ corresponds to a path in $G$
  with at most $h+\lceil h/3\rceil$ hops.

  We recursively solve the problem on the graph $G'$ with hop-bound $h'=\lceil 2h/3\rceil$.
  Let $\{\pi_t':t\in V\}$ be the obtained set of paths.
  For every $t\in V$, we define $\pi_t$ to be $\pi'_t$ with possibly the first shortcut edge $e_z$ expanded
  to the corresponding path $p_z$. 
  One can easily prove by induction that $|\pi_t|=O(h)$ and $\len(\pi_t)\leq \dist_G^h(s,t)$.
  The recursion depth is clearly $O(\log{h})$.

  Finally, each of the explicitly stored $\Ot(n/h)$ paths $p_t$ at some level of the recursion can be
  imagined to point to at most one path of the previous level (corresponding to a shortcut
  edge) and some $O(h)$ distinct vertices of $G$. By keeping only the nodes
  reachable from the paths at the last level of the recursion in this pointer
  system, and storing reverse pointers, we can report the elements of each
  $P_v$ so that every element gets reported $O(\log{n})$ times.
\end{proof}

To reduce the space to $\Ot(n^2)$, we simply replace the Bellman-Ford-like procedure run on $G-C$
in the preprocessing of Lemma~\ref{l:prep} with the construction of Lemma~\ref{l:space}.
The total congestion of all the vertices can increase only by a constant factor then.
In Section~\ref{s:deterministic} we have assumed that the preprocessed paths $\pi_{s,t}$
were simple when hitting all $\pi_{s,t}$ satisfying $|\pi_{s,t}|\geq h/2$ with~$H_0$.
But we can as well assume that $H_0$ hits all $\pi_{s,t}$ with $|V(\pi_{s,t})|\geq h/2$
instead. Even though the paths represented by the data structure of Lemma~\ref{l:space}
might be non-simple, we can compute the sizes $|V(\pi_v)|$ within the same bound easily.
Moreover, the algorithm behind Lemma~\ref{l:hitting} can be implemented so that it
requires only $O(n)$ additional
space if it is possible to (1) iterate through the elements of individual sets of $\mathcal{Y}$
in $O(1)$ time per element, and (2) report the sets $Y\in\mathcal{Y}$ containing a given $x\in X$
in near-linear time in the number of reported sets. This is precisely what Lemma~\ref{l:space} enables.

\subsection{Negative edges and cycles}\label{s:negative}
In this section, we briefly describe the modifications to the data structure needed to handle
negative edge weights and possibly negative cycles.

First of all, we run in parallel a deterministic fully dynamic negative cycle detection algorithm
with $\Ot(m)$ worst-case update time (see, e.g.,~\cite{Karczmarz21}).
That algorithm also maintains a feasible price function~$p$ of the current graph $G$.
With this in hand, whenever $G$ has a negative cycle, we refrain from running the update
procedure and forbid issuing queries. Otherwise, $p$ is also a feasible price function of $G-D\subseteq G$, and thus
the Dijkstra-based update procedure can simply use $p$ to ensure that all
the edge and path lengths accessed are non-negative.

In the basic randomized variant of our data structure we don't need to alter
the preprocessing at the beginning of a phase at all. Indeed, our basic analysis did not
require that the paths $\pi'_{s,t}$ are simple or with no negative cycles,
and $h$-hop-bounded shortest paths are well-defined even in presence of negative cycles.
In the $O(n^2h)$-space deterministic variant (Section~\ref{s:deterministic}), similarly as in Section~\ref{s:space},
we may compute the hitting set $H_0$ only for those $\pi_{s,t}$ that satisfy
$|V(\pi_{s,t})|\geq h/2$. Recall that if the update procedure is run, then $G-D$ has no negative cycle and hence no path $\pi_{s,t}$
containing a negative cycle survives in $G-D$ anyway.

Finally, the preprocessing algorithm behind Lemma~\ref{l:space} internally
uses hitting-set arguments (valid for simple paths) and requires, out-of-the-box, that there are no negative cycles.
We now sketch how to deal with negative cycles while using the space-saving Lemma~\ref{l:space}.

Whenever the preprocessing in Lemma~\ref{l:space}
for source $s$ encounters a path $p_t$ containing a negative cycle, we use it as the desired path $\pi_t$, but discard it when computing
a hitting set and thus also in the recursive preprocessing in Lemma~\ref{l:space} --
effectively making reporting $\pi_t$ (in any way) during update or query impossible. Similarly, such a path is included
as $\pi_{s,t}\in\Pi$ in Lemma~\ref{l:prep} only implicitly and marked as \emph{negative},
but nevertheless used for updating the congestion counters $\alpha(\cdot)$ during the preprocessing.
Note that during the update procedure, if~$G$ has no negative cycles, then for each ``negative'' path $\pi_{s,t}$,
we have $V(\pi_{s,t})\cap D\neq \emptyset$. The used charging scheme ensures
that we can afford reconstructing the path $\pi'_{s,t}$ within the $\Ot(\tau\Delta)$ bound even though we do not know
which vertices of $D$ lie on~$\pi_{s,t}$.

\section{Algebraic fully dynamic reachability in sparse DAGs}\label{sec:dag}
In this section we show how the algebraic approach to dynamic reachability~\cite{Sankowski04}
can be applied in the case of sparse DAGs,
even without resorting to fast matrix multiplication~\cite{WilliamsXXZ24}.

Assume for simplicity that $m=|E(G)|\geq n$ at all times. We prove the following.
\begin{restatable}{theorem}{sparsereach}\label{t:sparse-reach}
  Let $G$ be a directed acyclic graph. Let $t\in [1,\sqrt{m}]$. There exist a Monte Carlo randomized data structure maintaining
  $G$ subject to fully dynamic single-edge updates with $O(mn/t)$ worst-case update time and supporting
  arbitrary-pair reachability queries in $O(t)$ time. The answers produced are correct with high probability.
\end{restatable}

Let us first review the path-counting approach to dynamic DAG-reachability~\cite{DemetrescuI04, KingS02, Sankowski04}. Identify the vertices of $G=(V,E)$ with $\{1,\ldots,n\}$.
Let $A(G)$ be the adjacency matrix of $G$, that is, an $n\times n$ matrix with the entry $A(G)_{ij}$ equal to $1$
if $ij\in E(G)$, and $0$ otherwise. 
\begin{lemma}\label{l:matrix-inverse}{\upshape (e.g.,~\cite{Sankowski04})}
  If $G$ is a DAG  then the matrix $I-A(G)$ is invertible
  and for all $u,v\in V$, $(I-A(G))^{-1}_{u,v}$ equals the number of $u\to v$ paths in $G$.
\end{lemma}
\begin{proof}[Proof sketch]
  Let $A:=A(G)$. Recall that for any $k\geq 0$, the power $A^k$ encodes the pairwise counts of paths of length exactly $k$ in~$G$.
  In a DAG, all paths have length less than $n$ and thus $A^n$ is a zero matrix. As a result,
  $(I-A)\cdot (I+A^1+A^2+\ldots+A^{n-1})=I$. Thus $(I-A)^{-1}=I+A+\ldots+A^{n-1}$.
\end{proof}
In particular, by the above, testing whether a respective entry $(I-A(G))^{-1}_{s,t}$ is non-zero
gives the answer to a reachability query $(s,t)$ in $G$.
As a result, Lemma~\ref{l:matrix-inverse} reduces fully dynamic reachability to the \emph{dynamic matrix inverse} problem.
Specifically, a single-edge update to $G$ translates to a single-entry matrix update on $I-A(G)$,
whereas a reachability query corresponds to an element query on the inverse $(I-A(G))^{-1}$.
Moreover, as shown in~\cite{KingS02}, if one maintains the path counts modulo a random
prime number $p=\Theta(n^c)$ for a sufficiently large constant $c$, then high-probability
correctness is guaranteed over a polynomial number of updates. In other words,
it is enough to maintain the inverse $(I-A(G))^{-1}$ over the field $\field=\mathbb{Z}/p\mathbb{Z}$;
arithmetic operations on the field elements can be assumed to take unit time then.

Sankowski~\cite{Sankowski04}
studied update/query tradeoffs for the dynamic matrix inverse problem.
One tradeoff, summarized by the following theorem, is of particular interest here.
\begin{theorem}\label{t:inverse-tradeoff}{\upshape\cite{Sankowski04}}
  Suppose a matrix $A\in\field^{n\times n}$ is subject to single-element updates that keep~$A$ non-singular at all times.
  
  Let $\delta\in (0,1)$. There exists a data structure maintaining $A^{-1}$ with $\Ot(n^{\omega(1,\delta,1)-\delta}+n^{1+\delta})$
  worst-case update time and supporting element queries on $A^{-1}$ in $O(n^\delta)$ time.
\end{theorem}
Above, $\omega(1,\delta,1)\geq 2$ denotes the rectangular matrix multiplication exponent (see~\cite{GallU18}), i.e., a value
such that one can multiply an $n\times n^{\delta}$ matrix by an $n^{\delta}\times n$ matrix in $\Ot\left(n^{\omega(1,\delta,1)}\right)$ time.
Here, the time
is measured in field operations.
By applying Theorem~\ref{t:inverse-tradeoff} with $\delta\approx 0.529$ such that $\omega(1,\delta,1)=1+2\delta$
to the matrix $I-A(G)$, one
obtains a Monte Carlo randomized fully dynamic reachability algorithm for DAGs
with $\Ot(n^{1.529})$ worst-case update and $O(n^{0.529})$ query time, a trade-off
obtained first in~\cite{DemetrescuI04}.
Sankowski~\cite{Sankowski04} used the dynamic matrix inverse approach to general graphs as well;
however, more delicate arguments beyond path counting are needed then.

To continue, we need to discuss some of the internals of the data structure of Theorem~\ref{t:inverse-tradeoff}~\cite[Section~6]{Sankowski04}.
That data structures operates in phases of $n^\delta$ updates.
At the end of each phase, the inverse $A^{-1}$ is \emph{explicitly} recomputed
from (1) the explicitly stored inverse $(A_0)^{-1}$ of the matrix $A_0$ from the beginning of the phase,
and (2)~the~$n^\delta$ updates in the current phase, via rectangular matrix multiplication.
This is the sole reason why the term $n^{\omega(1,\delta,1)-\delta}$ appears in the update bound.
In particular, at the beginning of each phase, we could also recompute the inverse of
the current matrix~$A$ \emph{from scratch} in $O(n^\omega)$ time and thus obtain a slightly worse
update bound of $\Ot(n^{\omega-\delta}+n^{1+\delta})$, which in turn leads to the $\Ot(n^{(\omega+1)/2})=O(n^{1.687})$ update bound
if optimized wrt.~$\delta$.
The query time is proportional to the phase length $n^\delta$.

Speaking more generally, if we could explicitly recompute the maintained inverse
at any time in~$T$ time, then by following the approach behind~Theorem~\ref{t:inverse-tradeoff},
for any parameter $t\in [1,n]$ (denoting the phase length) we could obtain a data structure
with $\Ot(T/t+nt)$ worst-case update time and $O(t)$ query time.
For $T=\Omega(n)$, it only makes sense to use $t\in \left[1,\sqrt{T/n}\right]$, and the update
bound then simplifies to $\Ot(T/t)$.
To obtain a fully dynamic reachability algorithm for sparse DAGs this way,
we use this observation in combination with the below folklore lemma.

\begin{lemma}
  In a DAG $G$, all-pairs path counts modulo $p$ (and thus $(I-A(G))^{-1}$ over $\field=\mathbb{Z}/p\mathbb{Z}$ by Lemma~\ref{l:matrix-inverse}) can be computed
  in $O(mn)$ time.
\end{lemma}
\begin{proof}
  Fix $t\in V$. Let $c(u)$ be the number of $u\to t$ paths in $G$ modulo $p$. Then, $c(t)=1$ and for $u\neq t$ we have
  $c(u)=\left(\sum_{uv\in E}c(v)\right)\bmod{p}$. As a result, all the values $c(\cdot)$ can be computed
  inductively by considering the vertices in reverse topological order in time $O\left(\sum_{u\in V}\deg(u)\right)=O(m)$.
  The desired all-pairs counts are obtained by repeating this for all possible targets $t$.
\end{proof}
By the above lemma, for the case of the matrix $I-A(G)$ and acyclic $G$, we have $T=O(mn)$.
Theorem~\ref{t:sparse-reach} follows.

\bibliography{references}

\newcommand{\sortkey}[1]{}
\begin{thebibliography}{10}

\bibitem{AbboudW14}
Amir Abboud and Virginia~Vassilevska Williams.
\newblock Popular conjectures imply strong lower bounds for dynamic problems.
\newblock In {\em 55th {IEEE} Annual Symposium on Foundations of Computer
  Science, {FOCS} 2014}, pages 434--443. {IEEE} Computer Society, 2014.
\newblock \href {https://doi.org/10.1109/FOCS.2014.53}
  {\path{doi:10.1109/FOCS.2014.53}}.

\bibitem{AbrahamCK17}
Ittai Abraham, Shiri Chechik, and Sebastian Krinninger.
\newblock Fully dynamic all-pairs shortest paths with worst-case update-time
  revisited.
\newblock In {\em Proceedings of the Twenty-Eighth Annual {ACM-SIAM} Symposium
  on Discrete Algorithms, {SODA} 2017}, pages 440--452. {SIAM}, 2017.
\newblock \href {https://doi.org/10.1137/1.9781611974782.28}
  {\path{doi:10.1137/1.9781611974782.28}}.

\bibitem{AlokhinaB24}
Anastasiia Alokhina and Jan van~den Brand.
\newblock Fully dynamic shortest path reporting against an adaptive adversary.
\newblock In {\em Proceedings of the 2024 {ACM-SIAM} Symposium on Discrete
  Algorithms, {SODA} 2024}, pages 3027--3039. {SIAM}, 2024.
\newblock \href {https://doi.org/10.1137/1.9781611977912.108}
  {\path{doi:10.1137/1.9781611977912.108}}.

\bibitem{AlstrupHLT05}
Stephen Alstrup, Jacob Holm, Kristian de~Lichtenberg, and Mikkel Thorup.
\newblock Maintaining information in fully dynamic trees with top trees.
\newblock {\em {ACM} Trans. Algorithms}, 1(2):243--264, 2005.
\newblock \href {https://doi.org/10.1145/1103963.1103966}
  {\path{doi:10.1145/1103963.1103966}}.

\bibitem{AusielloIMN91}
Giorgio Ausiello, Giuseppe~F. Italiano, Alberto Marchetti{-}Spaccamela, and
  Umberto Nanni.
\newblock Incremental algorithms for minimal length paths.
\newblock {\em J. Algorithms}, 12(4):615--638, 1991.
\newblock \href {https://doi.org/10.1016/0196-6774(91)90036-X}
  {\path{doi:10.1016/0196-6774(91)90036-X}}.

\bibitem{BaswanaHS07}
Surender Baswana, Ramesh Hariharan, and Sandeep Sen.
\newblock Improved decremental algorithms for maintaining transitive closure
  and all-pairs shortest paths.
\newblock {\em J. Algorithms}, 62(2):74--92, 2007.
\newblock \href {https://doi.org/10.1016/j.jalgor.2004.08.004}
  {\path{doi:10.1016/j.jalgor.2004.08.004}}.

\bibitem{BaurS83}
Walter Baur and Volker Strassen.
\newblock The complexity of partial derivatives.
\newblock {\em Theor. Comput. Sci.}, 22:317--330, 1983.
\newblock \href {https://doi.org/10.1016/0304-3975(83)90110-X}
  {\path{doi:10.1016/0304-3975(83)90110-X}}.

\bibitem{BergamaschiHGWW21}
Thiago Bergamaschi, Monika Henzinger, Maximilian~Probst Gutenberg,
  Virginia~Vassilevska Williams, and Nicole Wein.
\newblock New techniques and fine-grained hardness for dynamic near-additive
  spanners.
\newblock In {\em Proceedings of the 2021 {ACM-SIAM} Symposium on Discrete
  Algorithms, {SODA} 2021}, pages 1836--1855. {SIAM}, 2021.
\newblock \href {https://doi.org/10.1137/1.9781611976465.110}
  {\path{doi:10.1137/1.9781611976465.110}}.

\bibitem{Bernstein16}
Aaron Bernstein.
\newblock Maintaining shortest paths under deletions in weighted directed
  graphs.
\newblock {\em {SIAM} J. Comput.}, 45(2):548--574, 2016.
\newblock \href {https://doi.org/10.1137/130938670}
  {\path{doi:10.1137/130938670}}.

\bibitem{BernsteinGS21}
Aaron Bernstein, Maximilian~Probst Gutenberg, and Thatchaphol Saranurak.
\newblock Deterministic decremental {SSSP} and approximate min-cost flow in
  almost-linear time.
\newblock In {\em 62nd {IEEE} Annual Symposium on Foundations of Computer
  Science, {FOCS} 2021}, pages 1000--1008. {IEEE}, 2021.
\newblock \href {https://doi.org/10.1109/FOCS52979.2021.00100}
  {\path{doi:10.1109/FOCS52979.2021.00100}}.

\bibitem{CasacubertaK22}
S{\'{\i}}lvia Casacuberta and Rasmus Kyng.
\newblock Faster sparse matrix inversion and rank computation in finite fields.
\newblock In {\em 13th Innovations in Theoretical Computer Science Conference,
  {ITCS} 2022}, volume 215 of {\em LIPIcs}, pages 33:1--33:24. Schloss Dagstuhl
  - Leibniz-Zentrum f{\"{u}}r Informatik, 2022.
\newblock \href {https://doi.org/10.4230/LIPIcs.ITCS.2022.33}
  {\path{doi:10.4230/LIPIcs.ITCS.2022.33}}.

\bibitem{ChechikZ23}
Shiri Chechik and Tianyi Zhang.
\newblock Faster deterministic worst-case fully dynamic all-pairs shortest
  paths via decremental hop-restricted shortest paths.
\newblock In {\em Proceedings of the 2023 {ACM-SIAM} Symposium on Discrete
  Algorithms, {SODA} 2023}, pages 87--99. {SIAM}, 2023.
\newblock \href {https://doi.org/10.1137/1.9781611977554.ch4}
  {\path{doi:10.1137/1.9781611977554.ch4}}.

\bibitem{ChenGHPS20}
Li~Chen, Gramoz Goranci, Monika Henzinger, Richard Peng, and Thatchaphol
  Saranurak.
\newblock Fast dynamic cuts, distances and effective resistances via vertex
  sparsifiers.
\newblock In {\em 61st {IEEE} Annual Symposium on Foundations of Computer
  Science, {FOCS} 2020}, pages 1135--1146. {IEEE}, 2020.
\newblock \href {https://doi.org/10.1109/FOCS46700.2020.00109}
  {\path{doi:10.1109/FOCS46700.2020.00109}}.

\bibitem{Chuzhoy21}
Julia Chuzhoy.
\newblock Decremental all-pairs shortest paths in deterministic near-linear
  time.
\newblock In {\em {STOC} '21: 53rd Annual {ACM} {SIGACT} Symposium on Theory of
  Computing}, pages 626--639. {ACM}, 2021.
\newblock \href {https://doi.org/10.1145/3406325.3451025}
  {\path{doi:10.1145/3406325.3451025}}.

\bibitem{ChuzhoyZ23}
Julia Chuzhoy and Ruimin Zhang.
\newblock A new deterministic algorithm for fully dynamic all-pairs shortest
  paths.
\newblock In {\em Proceedings of the 55th Annual {ACM} Symposium on Theory of
  Computing, {STOC} 2023}, pages 1159--1172. {ACM}, 2023.
\newblock \href {https://doi.org/10.1145/3564246.3585196}
  {\path{doi:10.1145/3564246.3585196}}.

\bibitem{DemetrescuI04}
Camil Demetrescu and Giuseppe~F. Italiano.
\newblock A new approach to dynamic all pairs shortest paths.
\newblock {\em J. {ACM}}, 51(6):968--992, 2004.
\newblock \href {https://doi.org/10.1145/1039488.1039492}
  {\path{doi:10.1145/1039488.1039492}}.

\bibitem{DemetrescuI05}
Camil Demetrescu and Giuseppe~F. Italiano.
\newblock Trade-offs for fully dynamic transitive closure on dags: breaking
  through the {O}(\({n}^{2}\)) barrier.
\newblock {\em J. {ACM}}, 52(2):147--156, 2005.
\newblock \href {https://doi.org/10.1145/1059513.1059514}
  {\path{doi:10.1145/1059513.1059514}}.

\bibitem{EvaldFGW21}
Jacob Evald, Viktor Fredslund{-}Hansen, Maximilian~Probst Gutenberg, and
  Christian Wulff{-}Nilsen.
\newblock Decremental {APSP} in unweighted digraphs versus an adaptive
  adversary.
\newblock In {\em 48th International Colloquium on Automata, Languages, and
  Programming, {ICALP} 2021}, volume 198 of {\em LIPIcs}, pages 64:1--64:20.
  Schloss Dagstuhl - Leibniz-Zentrum f{\"{u}}r Informatik, 2021.
\newblock \href {https://doi.org/10.4230/LIPIcs.ICALP.2021.64}
  {\path{doi:10.4230/LIPIcs.ICALP.2021.64}}.

\bibitem{ForsterGNS23}
Sebastian Forster, Gramoz Goranci, Yasamin Nazari, and Antonis Skarlatos.
\newblock Bootstrapping dynamic distance oracles.
\newblock In {\em 31st Annual European Symposium on Algorithms, {ESA} 2023},
  volume 274 of {\em LIPIcs}, pages 50:1--50:16. Schloss Dagstuhl -
  Leibniz-Zentrum f{\"{u}}r Informatik, 2023.
\newblock URL: \url{https://doi.org/10.4230/LIPIcs.ESA.2023.50}, \href
  {https://doi.org/10.4230/LIPICS.ESA.2023.50}
  {\path{doi:10.4230/LIPICS.ESA.2023.50}}.

\bibitem{ForsterNP22}
Sebastian Forster, Yasamin Nazari, and Maximilian~Probst Gutenberg.
\newblock Deterministic incremental {APSP} with polylogarithmic update time and
  stretch.
\newblock In {\em Proceedings of the 55th Annual {ACM} Symposium on Theory of
  Computing, {STOC} 2023}, pages 1173--1186. {ACM}, 2023.
\newblock \href {https://doi.org/10.1145/3564246.3585213}
  {\path{doi:10.1145/3564246.3585213}}.

\bibitem{GallU18}
Francois~Le Gall and Florent Urrutia.
\newblock Improved rectangular matrix multiplication using powers of the
  coppersmith-winograd tensor.
\newblock In {\em Proceedings of the Twenty-Ninth Annual {ACM-SIAM} Symposium
  on Discrete Algorithms, {SODA} 2018}, pages 1029--1046. {SIAM}, 2018.
\newblock \href {https://doi.org/10.1137/1.9781611975031.67}
  {\path{doi:10.1137/1.9781611975031.67}}.

\bibitem{GutenbergWW20}
Maximilian~Probst Gutenberg, Virginia~Vassilevska Williams, and Nicole Wein.
\newblock New algorithms and hardness for incremental single-source shortest
  paths in directed graphs.
\newblock In {\em Proccedings of the 52nd Annual {ACM} {SIGACT} Symposium on
  Theory of Computing, {STOC} 2020}, pages 153--166. {ACM}, 2020.
\newblock \href {https://doi.org/10.1145/3357713.3384236}
  {\path{doi:10.1145/3357713.3384236}}.

\bibitem{GutenbergW20b}
Maximilian~Probst Gutenberg and Christian Wulff{-}Nilsen.
\newblock Fully-dynamic all-pairs shortest paths: Improved worst-case time and
  space bounds.
\newblock In {\em Proceedings of the 2020 {ACM-SIAM} Symposium on Discrete
  Algorithms, {SODA} 2020}, pages 2562--2574. {SIAM}, 2020.
\newblock \href {https://doi.org/10.1137/1.9781611975994.156}
  {\path{doi:10.1137/1.9781611975994.156}}.

\bibitem{abs-2402-18541}
Bernhard Haeupler, Yaowei Long, and Thatchaphol Saranurak.
\newblock Dynamic deterministic constant-approximate distance oracles with
  n\({}^{\mbox{{\(\epsilon\)}}}\) worst-case update time.
\newblock {\em CoRR}, abs/2402.18541, 2024.
\newblock URL: \url{https://doi.org/10.48550/arXiv.2402.18541}, \href
  {http://arxiv.org/abs/2402.18541} {\path{arXiv:2402.18541}}, \href
  {https://doi.org/10.48550/ARXIV.2402.18541}
  {\path{doi:10.48550/ARXIV.2402.18541}}.

\bibitem{HenzingerKNS15}
Monika Henzinger, Sebastian Krinninger, Danupon Nanongkai, and Thatchaphol
  Saranurak.
\newblock Unifying and strengthening hardness for dynamic problems via the
  online matrix-vector multiplication conjecture.
\newblock In {\em Proceedings of the Forty-Seventh Annual {ACM} on Symposium on
  Theory of Computing, {STOC} 2015}, pages 21--30. {ACM}, 2015.
\newblock \href {https://doi.org/10.1145/2746539.2746609}
  {\path{doi:10.1145/2746539.2746609}}.

\bibitem{KaltofenP91}
Erich Kaltofen and Victor~Y. Pan.
\newblock Processor efficient parallel solution of linear systems over an
  abstract field.
\newblock In {\em Proceedings of the 3rd Annual {ACM} Symposium on Parallel
  Algorithms and Architectures, {SPAA} '91}, pages 180--191. {ACM}, 1991.
\newblock \href {https://doi.org/10.1145/113379.113396}
  {\path{doi:10.1145/113379.113396}}.

\bibitem{Karczmarz21}
Adam Karczmarz.
\newblock Fully dynamic algorithms for minimum weight cycle and related
  problems.
\newblock In {\em 48th International Colloquium on Automata, Languages, and
  Programming, {ICALP} 2021}, volume 198 of {\em LIPIcs}, pages 83:1--83:20.
  Schloss Dagstuhl - Leibniz-Zentrum f{\"{u}}r Informatik, 2021.
\newblock \href {https://doi.org/10.4230/LIPIcs.ICALP.2021.83}
  {\path{doi:10.4230/LIPIcs.ICALP.2021.83}}.

\bibitem{Karczmarz0S22}
Adam Karczmarz, Anish Mukherjee, and Piotr Sankowski.
\newblock Subquadratic dynamic path reporting in directed graphs against an
  adaptive adversary.
\newblock In {\em {STOC} '22: 54th Annual {ACM} {SIGACT} Symposium on Theory of
  Computing}, pages 1643--1656. {ACM}, 2022.
\newblock \href {https://doi.org/10.1145/3519935.3520058}
  {\path{doi:10.1145/3519935.3520058}}.

\bibitem{KarczmarzS23}
Adam Karczmarz and Piotr Sankowski.
\newblock Fully dynamic shortest paths and reachability in sparse digraphs.
\newblock In {\em 50th International Colloquium on Automata, Languages, and
  Programming, {ICALP} 2023}, volume 261 of {\em LIPIcs}, pages 84:1--84:20.
  Schloss Dagstuhl - Leibniz-Zentrum f{\"{u}}r Informatik, 2023.
\newblock URL: \url{https://doi.org/10.4230/LIPIcs.ICALP.2023.84}, \href
  {https://doi.org/10.4230/LIPICS.ICALP.2023.84}
  {\path{doi:10.4230/LIPICS.ICALP.2023.84}}.

\bibitem{King99}
Valerie King.
\newblock Fully dynamic algorithms for maintaining all-pairs shortest paths and
  transitive closure in digraphs.
\newblock In {\em 40th Annual Symposium on Foundations of Computer Science,
  {FOCS} 1999}, pages 81--91. {IEEE} Computer Society, 1999.
\newblock \href {https://doi.org/10.1109/SFFCS.1999.814580}
  {\path{doi:10.1109/SFFCS.1999.814580}}.

\bibitem{KingS02}
Valerie King and Garry Sagert.
\newblock A fully dynamic algorithm for maintaining the transitive closure.
\newblock {\em J. Comput. Syst. Sci.}, 65(1):150--167, 2002.
\newblock \href {https://doi.org/10.1006/jcss.2002.1883}
  {\path{doi:10.1006/jcss.2002.1883}}.

\bibitem{MakinenTKPGC19}
Veli M{\"{a}}kinen, Alexandru~I. Tomescu, Anna Kuosmanen, Topi Paavilainen,
  Travis Gagie, and Rayan Chikhi.
\newblock Sparse dynamic programming on dags with small width.
\newblock {\em {ACM} Trans. Algorithms}, 15(2):29:1--29:21, 2019.
\newblock \href {https://doi.org/10.1145/3301312} {\path{doi:10.1145/3301312}}.

\bibitem{Mao24a}
Xiao Mao.
\newblock Fully dynamic all-pairs shortest paths: Likely optimal worst-case
  update time.
\newblock In {\em Proceedings of the 56th Annual {ACM} Symposium on Theory of
  Computing, {STOC} 2024}, pages 1141--1152. {ACM}, 2024.
\newblock \href {https://doi.org/10.1145/3618260.3649695}
  {\path{doi:10.1145/3618260.3649695}}.

\bibitem{Roditty08}
Liam Roditty.
\newblock A faster and simpler fully dynamic transitive closure.
\newblock {\em {ACM} Trans. Algorithms}, 4(1):6:1--6:16, 2008.
\newblock \href {https://doi.org/10.1145/1328911.1328917}
  {\path{doi:10.1145/1328911.1328917}}.

\bibitem{RodittyZ08}
Liam Roditty and Uri Zwick.
\newblock Improved dynamic reachability algorithms for directed graphs.
\newblock {\em {SIAM} J. Comput.}, 37(5):1455--1471, 2008.
\newblock \href {https://doi.org/10.1137/060650271}
  {\path{doi:10.1137/060650271}}.

\bibitem{RodittyZ11}
Liam Roditty and Uri Zwick.
\newblock On dynamic shortest paths problems.
\newblock {\em Algorithmica}, 61(2):389--401, 2011.
\newblock \href {https://doi.org/10.1007/s00453-010-9401-5}
  {\path{doi:10.1007/s00453-010-9401-5}}.

\bibitem{Sankowski04}
Piotr Sankowski.
\newblock Dynamic transitive closure via dynamic matrix inverse (extended
  abstract).
\newblock In {\em 45th Symposium on Foundations of Computer Science {FOCS}
  2004}, pages 509--517. {IEEE} Computer Society, 2004.
\newblock \href {https://doi.org/10.1109/FOCS.2004.25}
  {\path{doi:10.1109/FOCS.2004.25}}.

\bibitem{Thorup04}
Mikkel Thorup.
\newblock Fully-dynamic all-pairs shortest paths: Faster and allowing negative
  cycles.
\newblock In {\em Algorithm Theory - {SWAT} 2004, 9th Scandinavian Workshop on
  Algorithm Theory, Proceedings}, volume 3111 of {\em Lecture Notes in Computer
  Science}, pages 384--396. Springer, 2004.
\newblock \href {https://doi.org/10.1007/978-3-540-27810-8\_33}
  {\path{doi:10.1007/978-3-540-27810-8\_33}}.

\bibitem{UY91}
Jeffrey~D. Ullman and Mihalis Yannakakis.
\newblock High-probability parallel transitive-closure algorithms.
\newblock {\em {SIAM} J. Comput.}, 20(1):100--125, 1991.
\newblock \href {https://doi.org/10.1137/0220006} {\path{doi:10.1137/0220006}}.

\bibitem{BrandFN22}
Jan van~den Brand, Sebastian Forster, and Yasamin Nazari.
\newblock Fast deterministic fully dynamic distance approximation.
\newblock In {\em 63rd {IEEE} Annual Symposium on Foundations of Computer
  Science, {FOCS} 2022}, pages 1011--1022. {IEEE}, 2022.
\newblock \href {https://doi.org/10.1109/FOCS54457.2022.00099}
  {\path{doi:10.1109/FOCS54457.2022.00099}}.

\bibitem{BrandN19}
Jan van~den Brand and Danupon Nanongkai.
\newblock Dynamic approximate shortest paths and beyond: Subquadratic and
  worst-case update time.
\newblock In {\em 60th {IEEE} Annual Symposium on Foundations of Computer
  Science, {FOCS} 2019}, pages 436--455. {IEEE} Computer Society, 2019.
\newblock \href {https://doi.org/10.1109/FOCS.2019.00035}
  {\path{doi:10.1109/FOCS.2019.00035}}.

\bibitem{BrandNS19}
Jan van~den Brand, Danupon Nanongkai, and Thatchaphol Saranurak.
\newblock Dynamic matrix inverse: Improved algorithms and matching conditional
  lower bounds.
\newblock In {\em 60th {IEEE} Annual Symposium on Foundations of Computer
  Science, {FOCS} 2019}, pages 456--480. {IEEE} Computer Society, 2019.
\newblock \href {https://doi.org/10.1109/FOCS.2019.00036}
  {\path{doi:10.1109/FOCS.2019.00036}}.

\bibitem{WilliamsXXZ24}
Virginia~Vassilevska Williams, Yinzhan Xu, Zixuan Xu, and Renfei Zhou.
\newblock New bounds for matrix multiplication: from alpha to omega.
\newblock In {\em Proceedings of the 2024 {ACM-SIAM} Symposium on Discrete
  Algorithms, {SODA} 2024}, pages 3792--3835. {SIAM}, 2024.
\newblock \href {https://doi.org/10.1137/1.9781611977912.134}
  {\path{doi:10.1137/1.9781611977912.134}}.

\bibitem{Zwick02}
Uri Zwick.
\newblock All pairs shortest paths using bridging sets and rectangular matrix
  multiplication.
\newblock {\em J. {ACM}}, 49(3):289--317, 2002.
\newblock \href {https://doi.org/10.1145/567112.567114}
  {\path{doi:10.1145/567112.567114}}.

\end{thebibliography}

\appendix

\section{Further variants of the fully dynamic shortest paths data structure}

\subsection{Unweighted digraphs}\label{s:unweighted}
Similarly as in the case of previous fully dynamic APSP data structures~\cite{AbrahamCK17, GutenbergW20b}, improved
bounds can be obtained if the graph $G$ is unweighted. This is simply because the preprocessing of Lemma~\ref{l:prep}
can be completed in $O(mn)$ time instead of $O(mnh)$ time.
Indeed, in an unweighted graph, the shortest $h$-hop-bounded $s,t$ path, if exists,
coincides with the (globally) shortest $s,t$ path. As a result, the Bellman-Ford-based computation
can be replaced with breadth-first search.
Similarly, the collection of paths $\Pi$ can be represented using $n$ BFS trees and thus
one can achieve quadratic space without resorting to Lemma~\ref{l:space}.

For unweighted graphs, the update bound becomes $\Ot(m\Delta+mn/\Delta+mn/h+nm^2h/\tau+\Delta\tau)$,
whereas the query time remains $\Ot(\Delta+nmh/\tau+n/h)$.
For $\Delta=h=n^{1/4}$ and $\tau=mn^{1/2}$ the update and query time bounds
become $\Ot(mn^{3/4})$ and $\Ot(n^{3/4})$, respectively.

\subsection{A slight tradeoff}\label{s:tradeoff}
In the basic variant of the data structure, it is not clear whether pushing the update time below $\Ot(n^{4/5})$ is possible even
at the cost of increasing the query time. Here, we sketch that a slight tradeoff is indeed possible
with another trick of~\cite[Section~4.1]{GutenbergW20b}: to delegate handling paths through the congested set
to the data structure of~\cite[Section~3]{AbrahamCK17}. For simplicity, assume again that the edge weights are non-negative.
Since that data structure, in turn, is tailored to dense
graphs, we instead use the following sparse variant implicit in~\cite{Karczmarz21}.
\begin{lemma}\label{l:ack}{\upshape\cite{AbrahamCK17, Karczmarz21}}
  Let $G=(V,E)$ be a directed graph and let $C\subseteq V$. Let $h\in [1,n]$.
  In $\Ot(|C|mh)$ time one can build a data structure supporting the following.

  For any query set ${D\subseteq V}$, update the data structure so that it supports
  queries computing the length of some $s\to t$ path of length at most $\min_{c\in C}\{\dist^h_{G-D}(s,c)+\dist^h_{G-D}(c,t)\}$ for any $s,t\in V$.
  The worst-case update time is $\Ot(|D|mh)$ and the query time is $O(|C|)$.
\end{lemma}
\begin{proof}[Proof sketch]
  For at most $2|C|$ \emph{centers} $c_1,\ldots,c_\ell$, repeatedly find shortest $2h$-hop-bounded paths from/to~$c_i$
  in $G-\{c_1,\ldots,c_{i-1}\}$. While this computation proceeds, maintain vertex congestions $\alpha(\cdot)$ as in Lemma~\ref{l:prep}.
  When choosing the subsequent centers $c_i$, alternate between picking an unused vertex from $C$
  and the most congested vertex of $V\setminus \{c_1,\ldots,c_{i-1}\}$,
  until all vertices of $C$ are used.
  This preprocessing costs $O(\ell mh)=O(|C|mh)$ time. 

  Given the above preprocessing, one can prove that by proceeding as in Lemma~\ref{l:dijkstra-rebuild}, in $\Ot(|D|mh)$ time
  one can recompute a representation of paths $s\to c_i$ of length at most $\dist_{G-(D\cup \{c_1,\ldots,c_{i-1}\})}^{2h}(s,c_i)$
  and analogous paths $c_i\to t$, for all $i$ and $s,t\in V$.

  Upon a query $(s,t)$, in $O(\ell)=O(|C|)$ time we can find an $s\to t$ path
  of length at most $y^*=\min_{i=1}^\ell\{y_i\}$, where $y_i:=\dist_{G-(D\cup \{c_1,\ldots,c_{i-1}\})}^{2h}(s,c_i)+\dist_{G-(D\cup \{c_1,\ldots,c_{i-1}\})}^{2h}(c_i,t)$.
  To see that this is enough, 
  let $c^*\in C$ be such that $\dist_{G-D}^h(s,c^*)+\dist_{G-D}^h(c^*,t)$ is minimum.
  Let $j$ be minimum index such that the corresponding $\leq 2h$-hop path $Q=s\to c^*\to t$ contains the center~$c_j$.
  Then we have $Q\subseteq G-(D\cup \{c_1,\ldots,c_{j-1}\})$ and thus $y^*\leq y_j\leq \len(Q)$.
\end{proof}

Note that by computing shortest-paths trees from and to a randomly sampled $\Ot(n/h)$-sized
hitting set~$H$ we can in fact handle ``long'' shortest
paths in the current graph $G$, and not only in $G-(C\cup D)$. As a result,
we don't need to recompute full shortest paths trees from $C$ -- instead, it would be enough
to consider short paths in $G-D$ through $C$ upon query.
This is what we use Lemma~\ref{l:ack} for. Every $\Delta$ updates, when a new
phase starts, a fresh congested set $C$ is computed.
We additionally initialize the data structure of Lemma~\ref{l:ack} for the current graph $G$
and the congested set $C$. This way, that data structure is always off from the current $G$
by at most $\Delta$ updates, and thus can be updated in $\Ot(\Delta mh)$ time. Again, the data structure of Lemma~\ref{l:ack} can be reinitialized
in such a way that the additional worst-case cost incurred is $\Ot(|C|mh/\Delta)$. The full worst-case
update time becomes:
\begin{equation*}
  \Ot(m\Delta+mnh/\Delta+mn/h+\Delta\tau+m^2nh^2/(\tau\Delta)+\Delta hm).
\end{equation*}
Balancing as before, for $\Delta=h^2$ and $\tau=mn/h^3$, we obtain the update bound
$\Ot(mn/h+mh^3)$.
Note that this bound is $\Omega(mn^{3/4})$ for any $h$.

The query bound unfortunately remains $\Ot(\Delta+|H|+mnh/\tau)=\Ot(n/h+h^4)$.
If we aim at serving $\Theta(n)$ queries per update and the graph is sparse, then we get
no improvement over the basic approach.
However, for a desired query time of $\Ot(t)$, where $t\in [n^{4/5},n]$, we can achieve
$\Ot(mn/t^{1/4})$ worst-case update time this way.

\end{document}